\documentclass[lettersize,journal]{IEEEtran}

\usepackage{tikz}
\newcommand*\emptycirc[1][0.85ex]{\tikz\draw (0,0) circle (#1);} 

\newcommand*\fullcirc[1][0.85ex]{\tikz\fill (0,0) circle (#1);}

\usepackage[lined,ruled,commentsnumbered]{algorithm2e}

\usepackage{cite}
\usepackage[numbers,sort&compress]{}
\usepackage{amsmath,amssymb,amsfonts}
\usepackage{array}
\usepackage{graphicx}
\usepackage{makecell} 

\usepackage{textcomp}
\usepackage{xcolor}
\usepackage{booktabs}
\usepackage{multirow}
\usepackage{subfigure}
\usepackage{longtable}
\usepackage{threeparttable}
\usepackage{enumitem}
\usepackage{fontawesome}
\usepackage{amssymb}
\usepackage{soul}
\usepackage{textcomp}
\usepackage{xcolor}
\usepackage{array}
\usepackage{url}
\usepackage{float}
\usepackage{bm}
\usepackage{adjustbox}
\usepackage{epstopdf}

\usepackage{diagbox}

\usepackage{amsthm}
\usepackage{tabularx}
\usepackage{mathrsfs} 
\usepackage{soul}
\usepackage{cuted}   
\usepackage{stfloats}

\theoremstyle{remark}  
\newtheorem{theorem}{Theorem}   

\usepackage{algpseudocode}

\usepackage{stfloats}
\usepackage[colorlinks=true, linkcolor=blue, citecolor=blue]{hyperref}

\usepackage{cleveref}
\begin{document}
\newcommand{\minitab}[2][l]{\begin{tabular}{#1}#2\end{tabular}}

\def\BibTeX{{\rm B\kern-.05em{\sc i\kern-.025em b}\kern-.08em
    T\kern-.1667em\lower.7ex\hbox{E}\kern-.125emX}}


\title{AmbShield: Enhancing Physical Layer Security with Ambient Backscatter Devices against Eavesdroppers}

\author{Yifan~Zhang,~\IEEEmembership{Graduate Student Member,~IEEE,}
        Yishan Yang,~\IEEEmembership{Graduate Student Member,~IEEE,}  \\ Masoud Kaveh,~\IEEEmembership{Member,~IEEE,}      
        Riku Jäntti,~\IEEEmembership{Senior Member,~IEEE,}
        Zheng~Yan,~\IEEEmembership{Fellow,~IEEE,} \\
        Dusit Niyato,~\IEEEmembership{Fellow,~IEEE,}
        and Zhu~Han,~\IEEEmembership{Fellow,~IEEE.}
\thanks{Y. F. Zhang, M. Kaveh, and R. Jäntti are with the Department of Information and Communications Engineering, Aalto University, Espoo, 02150, Finland. (email: yifan.1.zhang@aalto.fi, masoud.kaveh@aalto.fi, riku.jantti@aalto.fi)}
\thanks{Y. S. Yang and Z. Yan are with the School of Cyber Engineering, Xidian University, Xi'an, Shaanxi, 710026 China. (email: zyan@xidian.edu.cn, ysyangxd@stu.xidian.edu.cn)}%
\thanks{Dusit Niyato is with the College of Computing and Data Science, Nanyang Technological University,  639798, Singapore (e-mail: dniyato@ntu.edu.sg).}
\thanks{Z. Han is with the Department of Electrical and Computer Engineering, University of Houston, Houston, TX 77004 USA, and also with the Department of Computer Science and Engineering, Kyung Hee University, Seoul 446-701, South Korea (email: hanzhu22@gmail.com)}
 }

\markboth{Journal of \LaTeX\ Class Files}%
{Shell \MakeLowercase{\textit{et al.}}: A Sample Article Using IEEEtran.cls for IEEE Journals}
\maketitle

\begin{abstract}

Passive eavesdropping compromises confidentiality in wireless networks, especially in resource-constrained environments where heavyweight cryptography is impractical.
Physical layer security (PLS) exploits channel randomness and spatial selectivity to confine information to an intended receiver with modest overhead. However, typical PLS techniques, such as using beamforming, artificial noise, and reconfigurable intelligent surfaces, often involve added active power or specialized deployment, and, in many designs, rely on precise time synchronization and perfect CSI estimation, which limits their practicality.
Meanwhile, the role of ambient backscatter devices (AmBDs) in potentially strengthening the legitimate channel while limiting eavesdroppers in generalized wireless network settings has not been fully investigated.
To this end, we propose AmbShield, an AmBD–assisted PLS scheme that leverages naturally distributed AmBDs to simultaneously strengthen the legitimate channel and degrade eavesdroppers' without requiring extra transmit power and with minimal deployment overhead. In AmbShield, AmBDs are exploited as friendly jammers that randomly backscatter to create interference at eavesdroppers, and as passive relays that backscatter the desired signal to enhance the capacity of legitimate devices. We further develop a unified analytical framework that analyzes the exact probability density function (PDF) and cumulative distribution function (CDF) of legitimate and eavesdropper signal-to-interference-noise ratio (SINR), a closed-form secrecy outage probability (SOP), its high-SNR asymptote, and a secrecy diversity order (SDO). The analysis provides clear design guidelines on various practical system parameters to minimize SOP. Extensive experiments that include Monte Carlo simulations, theoretical derivations, and high-SNR asymptotic analysis demonstrate the security gains of AmbShield across diverse system parameters under imperfect synchronization and CSI estimation.

\end{abstract}


\section{Introduction} \label{s1_intro}


With the rapid evolution of 5G wireless communication technologies and beyond 6G, the demand for secure and reliable communication has grown significantly \cite{10044183}. Emerging wireless networks not only promise massive connectivity and ultra-high data rates, but also introduce increased vulnerability to passive eavesdropping attacks due to their inherent openness and broadcast nature \cite{zhang2025artificial}. Eavesdropping attacks pose a critical challenge due to their passive nature, which allows unauthorized entities to intercept confidential transmissions without being detected, severely undermining information security. Traditional cryptographic measures provide computational security, but often fail to protect the physical wireless medium due to limited device resources \cite{9524814} or potential key compromise \cite{li2024ai}. Consequently, there has been growing attention to physical layer security (PLS) \cite{10623395}, a complementary approach that takes advantage of the fundamental properties of wireless channels, such as fading, interference, and noise, to protect communication at the physical layer and achieve information-theoretic secrecy, providing an additional layer of protection against eavesdroppers.

Despite the advantages of PLS in wireless networks, existing schemes fall short in practice since it typically requires additional active power \cite{9350170,9349152}, costly hardware and deployment \cite{10227884,10012331,9439833,9763499,10534116}, and lack validation under imperfect synchronization and channel state information (CSI) estimation \cite{9350170,9349152,10227884,10012331,9439833,9763499,10534116,8737501,9745773}. For example, designs based on artificial noise (AN) \cite{9350170,9349152} deliberately inject a jam from a legitimate transmitter to deteriorate an eavesdropping link, increasing active power consumption, and often require sophisticated successive interference cancelation (SIC) at an intended receiver to eliminate the negative effect of AN. Similarly, beamforming PLS schemes \cite{10227884} increase the power consumption in weight control and the hardware overhead in the beam-forming circuitry. Although some schemes \cite{9439833,9763499,10534116} exploit economic and low-power-consumption reconfigurable intelligent surfaces (RISs) that shape propagation to enhance secrecy, their practical deployment is hindered by limited installation flexibility and high cost. Moreover, performance is highly sensitive to synchronization and CSI estimation errors, which can degrade secrecy, yet are rarely tested in design and validation \cite{8543573}.

In contrast, ambient backscatter devices (AmBDs) have recently drawn considerable attention for reliability and security purposes \cite{10130082}, due to their unique operational characteristics and broad applicability, and have even been incorporated into the 3GPP Release 19 \cite{3gppTR22840} study on ambient IoT. Most AmBDs operate passively by reflecting ambient radio frequency (RF) signals to transmit data \cite{yang2017modulation}, eliminating the need for active power supplies and significantly reducing the complexity and energy consumption of the hardware. In addition, their compact and simple design makes them extremely cost-effective, with unit costs of 7 to 15 cents (USD) \cite{8368232}, and a small size within a few square centimeters \cite{abdulghafor2021recent}. These distinctive features have made AmBDs widely adopted in applications such as conformal monitoring \cite{10130082} and logistics \cite{ahmed2024noma}, for ultra-low power data transmission or signal strength enhancement. 
Therefore, some previous studies have investigated PLS in an AmBD-aided network, using techniques such as AN \cite{8737501} and RIS \cite{10337761}, where AmBDs are conventional communication users. In addition, a work in \cite{10109228} first used AmBD as a friendly jammer to achieve PLS in a particular cooperative network.

Nevertheless, the role of AmBDs in PLS has not been fully invertigated. Most existing studies either treat AmBDs as communication users \cite{8737501,9745773,
10337761} that themselves require protection against eavesdropping or embed them in bespoke, scenario-specific architectures \cite{10109228} for secrecy enhancement, leaving their utility in general peer-to-peer links largely unexamined. In fact, the backscatter of an AmBD can act as a passive relay that strengthens the legitimate channels. Moreover, because AmBDs are passive reflectors, their backscattered signals can induce location-dependent interference at eavesdroppers while only minimally perturbing the intended receiver. By simultaneously improving legitimate signal quality and impairing eavesdropper reception, AmBDs offer a low-overhead avenue to enhance secrecy, positioning them as promising auxiliary enablers of practical PLS.


To fully facilitate the potential of AmBDs for security, this paper proposes AmbShield, a practical AmBD-assisted wireless system that leverages naturally distributed AmBDs in underlying environments to enhance PLS by simultaneously relaying legitimate transmission and jamming eavesdroppers. We first model the system, describe the function of AmBDs in the system, analyze the signal-to-interference ratio (SINR) for both legitimate and eavesdropping channels, and summarize the design goal of AmbShield. Then, we analyze the security performance of the proposed system by deriving the probability density function (PDF), the cumulative distribution function (CDF), and a closed-form secrecy outage probability (SOP). Furthermore, we analyze the asymptotic characters of the SOP in a high-SNR region and prove the secrecy diversity order (SDO) of one. Finally, we present a comprehensive performance evaluation that validates AmbShield’s security gains and corroborates our analytical results. The primary contributions of this paper are summarized as follows:

\begin{itemize}
\item We introduce AmbShield, a practical system that utilizes off-the-shelf AmBDs naturally distributed in real-world environments to enhance PLS against eavesdroppers. Specifically,  AmBDs act as friendly jammers that simultaneously relay legitimate transmission and jam eavesdroppers by randomly backscattering legitimate signals using a pre-set pattern. We develop a complete system model, derive SINR expressions for legitimate and wiretap links, and articulate the design goals of AmbShield.

\item We establish comprehensive performance analyzes for AmbShield: (i) the exact PDF and CDF of the legitimate SINR via Hankel transforms \cite{poularikas2018transforms} and Mellin–Barnes \cite{coelho2019finite} techniques with \cite{mathai2009h} Fox–H representations; (ii) the exact CDF and PDF of the eavesdropper SINR using Laplace transforms and exponential integration functions; (iii) a closed-form SOP in the Fox–H form and its high-SNR asymptotic analysis that exposes the secrecy diversity order, which is proved to be 1, ensuring linear secrecy improvement. These results yield tractable expressions that reveal how various system parameters, such as AmBDs, placement, and channel statistics, shape secrecy.

\item  We conduct a comprehensive performance evaluation of AmbShield through Monte Carlo simulations, theoretical analysis, and high-SNR asymptotic characterization. The simulation results closely follow the theoretical and asymptotic curves across a wide range of parameters, confirming both the accuracy and robustness of our analysis. Notably, AmbShield achieves an SOP below $10^{-3}$ under practical system settings, demonstrating its resilience to imperfections in both synchronization and CSI estimation. These findings highlight the strong potential of leveraging existing AmBDs to enhance PLS in real-world deployments.

\end{itemize}


\textbf{\textit{Remainder:}} The rest of the paper is organized as follows. Section~\ref{related work} reviews related work. Section~\ref{System model} introduces the AmbShield system model, including the channel and receiver models as well as design goals. Section~\ref{SNR Distribution} provides performance analysis, covering the SNR distribution, the Fox-H expression of SOP, and high-SNR asymptotics. Section~\ref{performance} presents numerical and Monte Carlo simulations to validate the analysis and examine the impact of key parameters. Finally, Section~\ref{conclusion} concludes the paper.

\textbf{\textit{Notations:}}
Though the paper is in italics $x$, in bold lowercase $\mathbf{x}$, and in uppercase $X$ denote scalars, vectors, and arrays, respectively.
$\mathbb{E}[\cdot]$ denotes expectation, $\Pr\{\cdot\}$ denotes probability, and $\mathcal{L}_X(s)$ denotes the Laplace transform of $X$.
$J_0(\cdot)$ and $K_0(\cdot)$ are Bessel functions of the first kind and modified of the second kind.
$E_1(\cdot)$ is the exponential integral.
$H^{m,n}_{p,q}[\cdot]$ is the Fox--$H$ function and $F^{(K)}_D(\cdot)$ is the Lauricella hypergeometric function.

\begin{table*}[]
\footnotesize
\centering
\caption{Comparison between AmbShield and related works.}
\label{related_work}
{\fullcirc: satisfy a criterion; \emptycirc: do not satisfy a criterion; -: not related.}
\\[2mm]

\begin{tabular}{c|c|c|c|c|c}
\toprule[1.5pt]
Reference & \minitab[c]{Main\\Technique} & \minitab[c]{AmBD\\Function} & \minitab[c]{Extra Power \\Avoidance} & \minitab[c]{Minimal   \\ Deployment Overhead} & \minitab[c]{Imperfect Sync. \&  CSI  \\  Estimation Evaluation} \\
\midrule

\cite{9350170} & AN & - & \emptycirc & \emptycirc & \emptycirc \\
\cite{9349152} & AN & - & \emptycirc & \fullcirc & \emptycirc \\
\cite{10227884,10012331}& Beamforming& - & \emptycirc & \emptycirc & \emptycirc \\
\cite{9439833,9763499,10534116}& RIS & - & \fullcirc & \emptycirc & \emptycirc \\
\cite{8737501,9745773} & AN &  user & \emptycirc & \emptycirc & \emptycirc \\
\cite{10337761}& RIS &  user & \fullcirc & \emptycirc & \emptycirc \\
\cite{10109228}& AmBD & friendly jammer & \fullcirc & \emptycirc & \emptycirc \\

\midrule

AmbShield & AmBD & friendly jammer & \fullcirc & \fullcirc & \fullcirc \\

\bottomrule[1.5pt]
\end{tabular}
\begin{tablenotes}
\item The results presented in the table are derived from the original results in the corresponding papers.
\end{tablenotes}
\vspace{-0.5cm}
\end{table*}

\section{Related Work} \label{related work}
In this section, we review state-of-the-art PLS schemes that leverage the typical technique against eavesdropping. Based on the main technique used in the scheme, we classify them into AN-based, beamforming-based, and RIS-based schemes. In addition, we review current PLS schemes related to AmBDs.
Table \ref{related_work} highlights the difference and novelty of AmbShield compared to existing PLS schemes in terms of main technique, AmBD function, extra power avoidance, minimal deployment overhead, and imperfect synchronization \& CSI estimation evaluation.

\vspace{-0.3cm}
\subsection{AN-based PLS Schemes}
 AN-based methods actively generate controlled jamming signals at legitimate transmitters to impair the eavesdropper's channel quality while minimally impacting the intended receiver. These methods often leverage sophisticated signal coding or SIC mechanisms to further enhance achievable secrecy performance. For example, Gong et al. \cite{9350170} employed a full-duplex user to emit AN to jam eavesdroppers for PLS in a large-scale downlink nonorthogonal multiple access (NOMA) network and derived the SOP and asymptotic diversity results using stochastic geometry.
They demonstrated that the AN-aided NOMA scheme achieves lower secrecy outage and higher secrecy throughput than orthogonal multiple access (OMA) and NOMA without AN. Although it is independent of eavesdropping CSI, it cannot avoid extra power due to AN transmission and deployment costs, due to the need for SIC mechanisms at legitimate receivers. Milad Tatar et al. \cite{9349152} proposed a wireless communication system able to support two-phase AN-aided secure unmanned aerial vehicle (UAV) and jointly optimized the UAV trajectory and AN power to maximize the average secrecy rate with clear gains over benchmarks. Their scheme leverages existing UAV infrastructure and employs pseudo-random AN, enabling straightforward cancellation at the legitimate receiver, thereby reducing deployment cost. However, similar to \cite{9350170}, they also need extra power for the AN transmission, and neither of the above two works \cite{9350170,9349152} considers imperfect synchronization and CSI estimation.

 \vspace{-0.3cm}
 \subsection{Beamforming-based PLS schemes}
 Beamforming-based PLS schemes, especially prevalent in multi-antenna or massive multiple-input multiple-output (massive-MIMO) systems, utilize directional beam steering to concentrate signals towards legitimate receivers, thereby significantly reducing information leakage to unauthorized parties. In recent years,  Li et al. \cite{10012331} formed a virtual antenna array from multiple UAVs and performed collaborative beamforming with joint positioning of the UAV, excitation weight design, and base station (BS) service ordering optimized by a parallelizable multi-objective dragonfly algorithm to maximize secrecy rate while suppressing sidelobes and limiting motion energy. However, they need additional overhead for multi-UAV synchronization and deployment. Su et al. \cite{10227884} proposed a sensing-assisted PLS framework where BS first estimates the directions of the eavesdroppers using a Capon-based detector and then jointly optimizes the AN-aided transmit beam formation based on the estimated result. They demonstrated that sensing and communication bring mutual gains in security. However, they need extra power for the AN and incur high overhead on the beamforming control and SIC. In addition, the impact of imperfect synchronization and CSI estimation was not evaluated in the above two works \cite{10012331,10227884}.

 \vspace{-0.3cm}
\subsection{RIS-based PLS schemes}
RIS-assisted security techniques have recently drawn considerable attention, particularly for next-generation communication networks. RIS devices serve as passive or active intelligent reflectors, dynamically reshaping the electromagnetic propagation environment. 
Zhang et al. \cite{9439833} analyzed RIS-aided downlink MIMO networks and derived closed-form SOP and ASR, which shows that passive beamforming RIS markedly improves PLS. Sun et al.  \cite{9763499} introduced a multilayer RIS–assisted transmitter
and an energy-efficient hybrid beamforming design under imperfect angular conditions to secure integrated terrestrial–aerial networks. Zhang et al. \cite{10534116} jointly optimized BS beamforming and RIS phase shifts to maximize the sum secrecy rate, providing optimal and low-complexity solutions and showing clear PLS gains for RIS-assisted two-way relay wiretap systems.
These schemes avoid extra power by using nearly passive elements in the RIS. Nevertheless, employing RIS introduces deployment and control costs, while imperfect synchronization and CSI estimation, which can severely undermine performance, remain largely unaddressed in the above works.


\vspace{-0.5cm}
\subsection{AmBD-related PLS schemes}

Recent studies have explored PLS in ambient backscatter–aided networks. Wang et al. \cite{8737501,9745773} employed AN to secure a backscatter network in which an AmBD acts as a communication user powered by a full duplex access point (AP) that simultaneously transmits energy-bearing signals and AN to prevent eavesdropping. However, this approach incurs extra power for the AN, additional overhead due to the SIC deployment, extra power for AN transmission, and assumes knowing perfect eavesdropping CSI. Kaveh et al. \cite{10337761} used RIS to secure an AmBD-to-AmBD link in smart grid scenarios, deriving SOP and average secrecy rate (ASR) in closed form via Fox-$H$ functions and analyzing their asymptotics. Notably, no extra power for signal transmission or eavesdropping CSI is required. 
The above works treat AmBDs as communication users that must be protected, while practical issues such as imperfect synchronization and imperfect CSI estimation are not explicitly addressed. In addition, closer to our AmbShield concept is \cite{10109228}, which proposes a two-stage ambient backscatter cooperative scheme in which AmBDs assist the legitimate link in the first slot and deliberately jam both the legitimate and the eavesdropping channels in the second. Although it is the first work to treat AmBDs as friendly jammers, it is confined to a specific cooperative architecture that requires an additional relay, and thus differs from AmbShield in both its system model and application scenarios. In addition, the impact of imperfect synchronization and imperfect CSI estimation remains unexplored. 

\textbf{\emph{Summary.}} As summarized in Table~\ref {related_work}, AmbShield leverages naturally distributed AmBDs in a dual role, both as passive relays to strengthen the legitimate link and as friendly interference to jam eavesdroppers, without additional transmit power and dedicated infrastructure. Importantly, AmbShield is evaluated under synchronization and CSI etimation imperfections, demonstrating practical robustness. These advantages highlight the practical potential of AmbShield.

\begin{figure}[tp]
    \centering
    \includegraphics[width=0.95\linewidth]{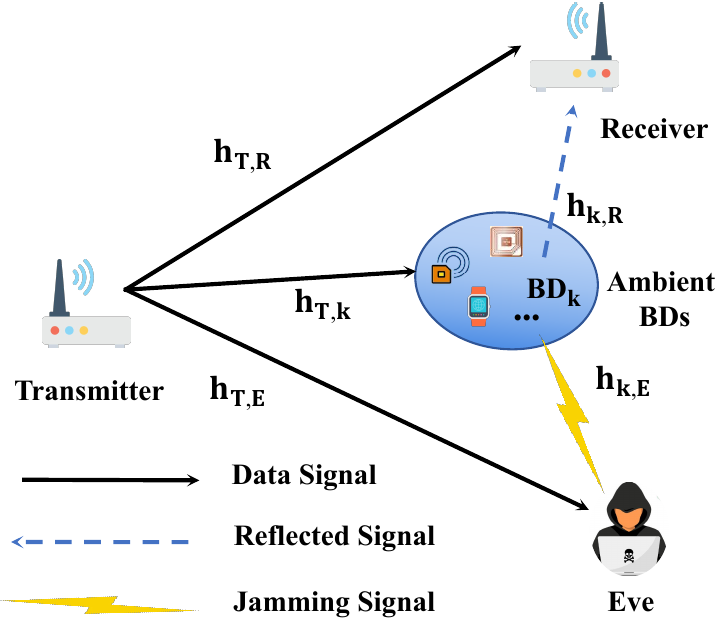}
    \caption{System model of AmbShield, where AmBDs are leveraged as friendly jammers that randomly reflect signals to simultaneously enhance communication performance and jam eavesdroppers. }
    \label{syse}
   \vspace{-0.5cm}
\end{figure}
\vspace{-0.1cm}
\section{System Model} \label{System model}
Fig. \ref{syse} shows the system model of the proposed system, where a transmitter transmits data signals to communicate with a receiver. Meanwhile, we consider that there is a passive eavesdropper named Eve who tries to decode the confidential message sent by the transmitter.  
The system also comprises $K$ ambient BDs distributed around the receiver. They randomly backscatter the incident signal from the transmitter towards the receiver and the eavesdropper.
Following common practice in AmBD-aided PLS analyses and prototypes \cite{10337761,9319204}, we consider a single antenna model (SISO) for all nodes to isolate the effect of AmBD on secrecy. The derivations extend to MIMO by replacing scalar channels with effective post-combining gains, leaving the qualitative insights unchanged \cite{oggier2011secrecy}.

\vspace{-0.3cm}
\subsection{Channel Model}
Denote $h_{t,r}$, $h_{t,k}$, and $h_{k,r}$ as the channels from the transmitter to the receiver, from the transmitter to the $k$-th BD, and from the $k$-th  BD to the receiver, respectively. ${h}_{t,e}$ and 
${h}_{k,e}$ are the channels from the transmitter to the eavesdropper and from the $k$-th  BD to the eavesdropper, respectively. The proposed system is applied in fields where some obstacles block the line-of-sight channel between the transmitter, the receiver, the BD, and Eve, such as logistics and industries \cite{10556582,ahmed2024noma}. Thus, all related channels between device $i$ and $j$ where $\{i,j\} \in \{t,k,e,r\}$ are modeled by an independent complex Gaussian random variable (Rayleigh fading) \cite{5559377} subject to large-scale fading as follows:

\begin{equation}
h_{i,j}
= \sqrt{ d_{ij}^{-\alpha}} \tilde h_{ij},
\end{equation}
where \(\alpha\) is the path loss exponent and \(d_{ij}\) is the distance between the device $i$ and the device $j$. Moreover, $h_{ij}  \sim \mathcal{CN} (0,\Omega_{ij} )$, where $\Omega_{ij} = d_{ij}^{-\alpha}$ and $\tilde h_{ij}$ represent the nonlinear sight component (NLoS), whose entries are a standard independent
and identically distributed complex Gaussian distribution (i.i.d.), i.e., $\tilde h_{ij} \sim \mathcal{CN}(0,1)$.

\vspace{-0.3cm}
\subsection{Receiver Model}
During a communication period, both the receiver and Eve will receive the data signals directly from the transmitter and the signal backscattered by BDs. Thus, the signal received by the receiver and the eavesdropper can be written as

\begin{align}
    y_r &= \sqrt{P}S(t)\left( h_{t,r} + \sum_{k=1}^K \phi_k h_{k,r} h_{t,k} \right)   + n_r, \\
    y_{e} &=  \sqrt{P}S(t)\left( h_{t,e} + \sum_{k=1}^K \phi_k h_{k,e} h_{t,k} \right)  + n_e,
\end{align}
where $S(t)$ is the signal transmitted to the transmitter. $
n_u \sim \mathcal{CN}(0, \sigma_r^2)$ and $n_e \sim \mathcal{CN} (0,\sigma_e^2  )$, are the AWGN \cite{tse2005fundamentals} at the receiver and the Eve, respectively. $P$ is the average transmit power. $\phi_k$ denotes the backscatter modulation of the $k$-th on the incident signal, which can be expressed as

\begin{equation}
\phi_k = a_k e^{j\varphi_k},
\end{equation}
where $a_k$ and $\varphi_k$ represent its amplitude and phase change, respectively. Considering the practical amplification ability of the amplifier, the amplitude is restricted by the maximum gain, that is, $a_k \le a_{\max},\ \forall k$. The BDs randomly backscatter the signal based on a pre-set pattern or seed, which is pre-known by legitimate receivers, and they demodulate or filter the backscatter signal. Meanwhile, the eavesdropper cannot obtain the backscatter modulations of BDs due to the lack of pre-set knowledge \cite{10599122}.

Then, the legitimate receiver distinguishes the signal from different BDs and filters the backscatter information $\phi_k$ to obtain the data signal from the backscattered signals. The decoded data signal from the backscatter signal is used to enhance communication. The instantaneous SINR at the receiver can be presented as 

\vspace{-0.3cm}
\begin{equation}\label{eq:gammaDef}
  \gamma_{r} =  \frac{P}{\sigma_{r}^{2}} | ( h_{t,r}  +  \sum_{k=1}^{K} \phi_{k} h_{k,r} h_{t,k}  ) |^{2}.
\end{equation}

Unlike the legitimate receiver, the backscatter modulation of the $K$ BDs is unknown to Eve, which it the interference. Hence, the instantaneous SINR at Eve is

\begin{equation}\label{eq:gamma_Eve_vec_SISO}
  \gamma_{e}
   = 
  \frac{P |h_{t,e}|^{2}}
       {\displaystyle
        \sigma_e^{2}+
        P\sum_{k=1}^{K}
          a_k^{2} 
          |h_{t,k}|^{2} |h_{k,e}|^{2}}.
\end{equation}

Compared \eqref{eq:gammaDef} with \eqref{eq:gamma_Eve_vec_SISO}, a performance difference can be found between the legitimate receiver and Eve, which leads to the potential secrecy capacity.

\vspace{-0.3cm}
\subsection{Design Goals}

In this paper, we propose a practical and secure AmBD-assisted wireless system that leverages off-the-shelf AmBDs naturally present in the environment to protect communications. To ensure practicality, the design should not introduce additional power, hardware, or dedicated deployment overhead. It should also be robust to synchronization and channel estimation errors. In addition, our goal is to analyze the security of AmbShield by obtaining key metrics, including the SOP and the SDO. The next section details the corresponding performance analysis.


\vspace{-0.1cm}
\section{Performance Analysis} \label{SNR Distribution}


\subsection{PDF and CDF of AmbShield}

\subsubsection{Legitimate Channel}

Let $X = ( h_{t,r}  +  \sum_{k=1}^{K} \phi_{k} h_{k,r} h_{t,k} )$,  the instantaneous legitimate SINR in \eqref{eq:gammaDef} is expressed as 
\begin{equation}
  \gamma_{r} = \frac{P}{\sigma_{r}^{2}} |X |^{2}.
\end{equation}

Since the channel response of the coherent overlay $X$, its PDF cannot be calculated directly using a standard function.
Therefore, we derive its PDF by separately calculating the zero-order Hankel transform \cite{poularikas2018transforms} for direct and BD cascade links. For the direct link, let $U=\lvert h_{t,r}\rvert^{2}$ and $h_{t,r} \sim \mathcal{CN}(0, \Omega_{tr})$, we thus have its zero-order Hankel transform:

\begin{align}
\Phi_{U}(t)&=\int_{0}^{\infty}f_{U}(z) J_{0}   (2\sqrt{tz} ) \mathrm dz  \nonumber  \\
&=\int_{0}^{\infty}
      \frac{1}{\Omega_{tr}} e^{-z/\Omega_{tr}}
       J_{0}   (2\sqrt{t z} ) \mathrm{d}z,\label{PhiU}
\end{align}
where $f_{U}(z)
    = \frac{1}{\Omega_{tr}} e^{-z/\Omega_{tr}}$ is derived since $U$ obeys the exponential distribution with the parameter $\frac{1}{\Omega_{tr}}$.

According to the Weber-Sonine integral \cite{watson1922treatise}, we have$\int_{0}^{\infty}   e^{-z/\Omega_{tr}}  J_{0}   (2\sqrt{t z} ) \mathrm{d}z = \Omega_{tr} e^{-t \Omega_{tr}}$. Thus, \eqref{PhiU} can be further calculated as 
\begin{equation}
\Phi_{U}(t) = \frac1{\Omega_{tr}} (\Omega_{tr} e^{-t \Omega_{tr}})=e^{-t \Omega_{tr}}
\end{equation}

For the BD cascade link, let $V_{k}=|\phi_{k}h_{k,r}h_{t,k}|^{2}$, the zero‑order Hankel transform is calculated as
\begin{equation} \label{eq:HandoublePDF}
\Phi_{V_k}(t) 
               =\int_{0}^{\infty} f_{V_k}(v) 
                 J_0   (2\sqrt{t v} ) \mathrm dv.
\end{equation}

Since $h_{k,r}$ and $h_{t,k}$ are independent of each other and $\lvert h_{k,r}\rvert^{2}\sim\mathrm{Exp}   (\tfrac{1}{\Omega_{k,r}} ), \lvert h_{t,k}\rvert^{2}\sim\mathrm{Exp}   (\tfrac{1}{\Omega_{t,k}} )$, $V_{k}$ obeys the double-Rayleigh distribution, whose PDF can be calculated as

\vspace{-0.3cm}
\begin{equation}\label{eq:doublePDF}
  f_{V_{k}}(v)=\frac{2}{\beta_{k}} K_{0}   (2\sqrt{v/\beta_{k}} ),
\end{equation}
where $\beta_k = a_k^2   \Omega_{t,k}   \Omega_{k,r}$ is the cascaded power gain of the $k$‑th BD, and the proof is derived in Appendix \ref{Appendx: V_k}.

Substituting \eqref{eq:doublePDF} into \eqref{eq:HandoublePDF} gives

\begin{equation}
\Phi_{V_k}(t)
= \frac{2}{\beta_k}
\int_{0}^{\infty}
K_0 (2\sqrt{\frac{\nu}{\beta_k}} ) 
J_0 (2\sqrt{t\nu} ) 
\mathrm{d}\nu. \label{HanV}
\end{equation}

Let  $x = \sqrt{\frac{v}{\beta_k}}$, so we have $v=\beta_k x^2$, and $\mathrm{d}v=2\beta_k x \mathrm{d}x$. Now, \eqref{HanV} can be simplified to 

\begin{equation}
\Phi_{V_k}(t)
=4\int_{0}^{\infty}
x K_0(2x) 
J_0 (2x\sqrt{\beta_k t} ) 
\mathrm{d}x.
\end{equation}

According to Weber–Schafheitlin integral \cite{gradshteyn2014table} $\int_{0}^{\infty}
x K_0(2a x) 
J_0(2b x) 
\mathrm{d}x
=\frac{1}{4 (a^2+b^2 )},
\qquad \Re\{a\}>0$ and let $a=1,\quad b=\sqrt\beta_k t$, we have:

\begin{equation}
 \Phi_{V_{k}}(t)= (1+\beta_{k}t )^{-1},
\end{equation}

For zero-order Hankel transform, we have $\Phi_{|Z_1+Z_2|^2}(t)=\Phi_{|Z_1|^2}(t) \Phi_{|Z_2|^2}(t)$, where $Z$ denotes a complex random quantities symmetric to a circle. Therefore, $\Phi_{X}(t)$ can be calculated though

\vspace{-0.3cm}
\begin{equation}\label{eq:PhiX}
  \Phi_{X}(t)=\Phi_{U}(t)\prod_{k=1}^{K}\Phi_{V_k}(t)= e^{-\Omega_{tr}t}\prod_{k=1}^{K} (1+\beta_{k}t )^{-1}.
\end{equation}

\textit{Remark 1} (Fox–$H$ function \cite{coelho2019finite,mathai2009h}): 
For integers $0\le n\le p$ and $0\le m\le q$, with parameters
$\{a_i,A_i\}_{i=1}^p$ and $\{b_j,B_j\}_{j=1}^q$ satisfying $A_i>0$, $B_j>0$,
the Fox--$H$ function is defined by the Mellin--Barnes integral in \eqref{eq:FoxH_def}. The Fox--$H$ expressions in established special functions constitute closed forms, which could be considered as closed form.
Since it is difficult to obtain a uniform and compact expression for PDF using elementary functions or finite-term hypergeometric functions based on the coherently overlapped SINR in \eqref{eq:gammaDef}, Fox--$H$ function is used in this manuscript to unify all intermediate transforms, such as Hankel and Laplace, into a single special-function layer that cleanly encodes the poles contributed by the direct path and each BD cascade.

\begin{figure*}[!t]
\begin{equation}
\label{eq:FoxH_def}
H^{m,n}_{p,q}\!\left[z\; |\;
\begin{matrix}
(a_1,A_1),\ldots,(a_p,A_p)\\
(b_1,B_1),\ldots,(b_q,B_q)
\end{matrix}
\right]
= \frac{1}{2\pi i}\int_{\mathcal{L}}
\frac{\prod_{j=1}^{m}\Gamma(b_j+B_j s)\,
      \prod_{i=1}^{n}\Gamma(1-a_i-A_i s)}
     {\prod_{j=m+1}^{q}\Gamma(1-b_j-B_j s)\,
      \prod_{i=n+1}^{p}\Gamma(a_i+A_i s)}
\,z^{-s}ds.
\end{equation}
\hrulefill
\end{figure*}

Based on the above results, the closed-form Fox–$H$ representation for the PDF of the legitimate channel can be expressed as
\begin{align}
f_{\gamma_r}(\gamma)
&=\frac{\sigma_r^2}{P} f_W   (\tfrac{\sigma_r^2}{P}\gamma )\nonumber
 \\
&=\frac{\sigma_r^2 }{\Omega_{tr}P} 
H_{1,K+1}^{K+1,1}  \left[
\frac{\sigma_r^2\gamma \Omega_{tr}^{ K-1}}{P\prod_k\beta_k}
  | 
\begin{array}{c}
(1,1)\\[2pt]   \underbrace{(0,1),\ldots,(0,1)}_{K+1}
\end{array}
\right],
\label{eq:sinrPDF}
\end{align}
where the detailed proof for this theorem is presented in Appendix \ref{APPEN_LPDF}.






Then, the CDF of the legitimate channel can be obtained through:

\vspace{-0.3cm}
\begin{equation}
F_{\gamma_r}(\gamma)=\int_{0}^{\gamma} f_{\gamma_r}(t) dt .
\end{equation}

We use the standard Fox--$H$ integral identity \cite{mathai2009h}:

\begin{equation}
\int_{0}^{x} H_{p,q}^{m,n}   [a t ] dt
= x  H_{p+1,q+1}^{m,n+1}  \left[
a x  | 
\begin{array}{c}
(0,1), (a_i,A_i)_{1:p}\\[2pt]
(b_j,B_j)_{1:q}, (0,1)
\end{array}
\right].
\end{equation}

Here $p=1, q=K+1, m=K+1, n=1$. Applying the integral identity, the CDF of the legitimate channel gives

\begin{equation}
F_{\gamma_r}(\gamma) \hspace{-2pt}
= \hspace{-2pt}
\frac{\sigma_r^2\gamma}{\Omega_{tr}P} 
H_{2,K+2}^{ K+1, 2} \hspace{-2pt} \left[
 \frac{\sigma_r^2\gamma \Omega_{tr}^{ K-1}}{P\prod_k\beta_k}   | 
\begin{array}{c}
(0,1), (1,1)\\
\underbrace{(0,1),\ldots,(0,1), (0,1)}_{K+2}
\end{array} \hspace{-4pt}
\right] \hspace{-2pt}.
\label{eq:grmmar}
\end{equation}

\subsubsection{Eavesdropping Channel }\label{Eves_SISO}
Define auxiliary variable $S = P |h_{t,e}|^{2}, I  =  P\sum_{k=1}^{K} a_k^{2}|h_{t,k}|^{2}|h_{k,e}|^{2}$, and $N = \sigma_e^{2}$,  \eqref{eq:gamma_Eve_vec_SISO} now becomes

\vspace{-0.3cm}
\begin{equation}
\gamma_e=\frac{S}{N+I}.
\end{equation}

Under independent Rayleigh fading of $S$ and $\sigma_e^{2}+I$, the PDF and the CDF of $\gamma_e$ can be calculated through a conditional expectation. 
Specifically, for Rayleigh fading $h_{i,j}$, $|h_{i,j}|^{2}$ is exponentially distributed with mean $\Omega_{i,j}=\mathbb{E}[|h_{i,j}|^2]$. Consequently, $S$ is exponential with mean $P\Omega_{t,e}$.

Conditioned on $I$, the CDF of $\gamma_e$ can be calculated though

\begin{align}
F_{\gamma_e|I}(x)
  &= \Pr  \left(\frac{S}{N+I}\le x  \Big|  I\right)
   = \Pr  \left(S \le x(N+I) \Big|  I\right) \notag\\
  &= 1-\exp  \left(-\frac{x(N+I)}{P\Omega_{t,e}}\right).
\end{align}

Averaging over $I$, the CDF of the eavesdropping channel can be calculated as
\begin{equation}\label{eq:CDF_generic}
F_{\gamma_e}(x)
= 1-\exp  \left(-\frac{x\sigma_e^2}{P\Omega_{t,e}}\right) 
\mathcal{L}_I  \left(\frac{x}{P\Omega_{t,e}}\right),
\end{equation}
where $\mathcal{L}_I(s)=\mathbb{E}  \left[e^{-sI}\right]$ is the Laplace transform of $I$.

Now we consider a single term $W_k = |h_{t,k}|^{2}|h_{k,e}|^{2}=X_kY_k$. Since
$X_k\sim\text{Exp}(1/\Omega_{t,k})$ and
$Y_k\sim\text{Exp}(1/\Omega_{k,e})$ are independent, $\mathcal{L}_{W_k}$ can be calculated though a Laplace transform:

\begin{align}
\mathcal{L}_{W_k}(s)
&=\int_{0}^{\infty} e^{-sw} f_{W_k}(w) dw \notag\\
&=\frac{1}{s\Omega_{t,k}\Omega_{k,e}}
   \exp  \left(\frac{1}{s\Omega_{t,k}\Omega_{k,e}}\right)
   E_1  \left(\frac{1}{s\Omega_{t,k}\Omega_{k,e}}\right),
\label{eq:LWk}
\end{align}
with $E_1(z)=\displaystyle\int_{z}^{\infty}\frac{e^{-t}}{t} dt$ denotes the exponential integral. Since $I=P\sum_{k=1}^{K} a_k^{2}W_k$, $\mathcal{L}_I(s)$ can be calculated in \eqref{eq:LI_final}.

\begin{figure*}[t]
			\normalsize
\begin{equation}
\mathcal{L}_I(s)=\prod_{k=1}^{K} \mathcal{L}_{W_k}(s P a_k^{2})=\prod_{k=1}^{K}
\frac{1}{s P a_k^{2}\Omega_{t,k}\Omega_{k,e}}
\exp  \left(\frac{1}{s P a_k^{2}\Omega_{t,k}\Omega_{k,e}}\right)
E_1  \left(\frac{1}{s P a_k^{2}\Omega_{t,k}\Omega_{k,e}}\right).
\label{eq:LI_final}
\end{equation}
\hrulefill
\vspace{-8pt}
\end{figure*}

Substitute $s=\dfrac{x}{P\Omega_{t,e}}$ from \eqref{eq:CDF_generic} into \eqref{eq:LI_final}, the CDF of the eavesdropping channel can be obtained in \eqref{eq:CDF_closed}.

\begin{figure*}[t]
			\normalsize
\begin{equation}
F_{\gamma_e}(x)
=1-\exp  \left(-\frac{x\sigma_e^2}{P\Omega_{t,e}}\right)
\prod_{k=1}^{K}
\frac{
\exp  \left(\dfrac{\Omega_{t,e}}{x P a_k^{2}\Omega_{t,k}\Omega_{k,e}}\right)
E_1  \left(\dfrac{\Omega_{t,e}}{x P a_k^{2}\Omega_{t,k}\Omega_{k,e}}\right)
}{
\dfrac{x P a_k^{2}\Omega_{t,k}\Omega_{k,e}}{\Omega_{t,e}}
}.
\label{eq:CDF_closed}
\end{equation}
\hrulefill
\vspace{-8pt}
\end{figure*}


Differentiating the CDF, the PDF of the eavesdropping channel yields

\vspace{-0.3cm}
\begin{align}
f_{\gamma_e}(x)= 
&\lambda 
 e^{-\lambda\sigma_e^{2}x} 
 \mathcal L_I(\lambda x)
\nonumber \\[2pt]
&\times 
 \Bigl[
   \sigma_{e}^{2}
   +\frac{P\Omega_{t,e}}{x}
     \sum_{k=1}^{K} 
       \Bigl(
         1+u_k(x)
         -\tfrac{e^{-u_k(x)}}{E_1 \bigl(u_k(x)\bigr)}
       \Bigr)
 \Bigr],  \label{eq: EPDF}
\end{align}
where $\lambda    = \frac{1}{P\Omega_{t,e}},
u_k(x) = \frac{P\Omega_{t,e}}{x A_k}$, and the derivation is summarized in Appendix \ref{APPEN_EPDF}.

\vspace{-0.1cm}
\subsection{Closed-form Secrecy Outage Probability}\label{sec:SOP}

With instantaneous SNRs \(\gamma_r\) and
\(\gamma_e\), the instantaneous secrecy capacity (SC) is

\begin{equation}
  C_s=[\log_2(1+\gamma_r)-\log_2(1+\gamma_e)]^{+}.
\end{equation}

As an information-theoretical metric to evaluate the PLS performance, SOP is defined as the probability that SC is less than a target secrecy rate threshold \(R_s   [{\rm bit/s/Hz}]\). Define the rate threshold $\Theta   =  2^{R_s} > 1$, the secrecy outage occurs whenever \(C_s<R_s\), i.e.,

\begin{equation}\label{eq:SOP_def}
  P_{\rm SOP}  = 
  \Pr    [\gamma_r<\Theta(1+\gamma_e)-1 ].
\end{equation}

Using the CDF of \(\gamma_r\) and the PDF of \(\gamma_e\), the SOP expression is given in the following theorem.


\begin{theorem}\label{}
The Fox-H SOP for the AmbShield is

\begin{equation}
P_{\mathrm{SOP}}
= \frac{c}{\Omega_{t,r}  \sigma_e^2}
 H_{2K+3, 0}^{0, 2K+3}
 [
\frac{d}{A_r c} \sigma_e^{-2}
 |(0,1), \mathcal{A}
 ],
\label{eq:SOP_final_corrected}
\end{equation}
where $ A_{r}
= \frac{\sigma_{r}^{2} \Omega_{t,r}^{ K-1}}
{P\prod_{k=1}^{K}\beta_{k}}, 
c = \Theta - 1, 
d = \frac{\Theta}{c}, 
\lambda = \frac{1}{P \Omega_{t,e}}, 
\mathcal{A} = [(0,1)^{K}, (0,\tfrac12)^{K}, (0,1), (1,1)] $.

\end{theorem}

\begin{proof}
The detailed proof of is in Appendix \ref{APPEN_SOP}.

\end{proof}


\vspace{-0.7cm}
\subsection{High‑SNR Asymptotic Analysis of the SOP}
To complement the exact SOP analysis, we further investigate the high-SNR asymptotics. This approach not only simplifies the performance characterization, but also reveals the fundamental scaling behavior and limits, providing clear insights into AmbShield’s secrecy performance under large transmit power. The asymptotic SOP is derived in the following theorm.

\label{sec:Asymptotic_SOP}

\begin{theorem} \label{corollary1}
The asymptotic SOP for AmbShield is given by \eqref{Asm_finnal}.

\begin{figure*}[t]
			\normalsize
\begin{equation}
P^{\rm Asy}_{\mathrm{SOP}}
=\frac{\sigma_r^{2}}{P \Omega_{t,r}} 
   [\Theta-1+\Theta P\Omega_{t,e} 
  H^{0, 2K+2}_{ 2K+2, 0} 
  \left[
     \sigma_e^{-2}
       | 
     (0,1)^{K}, (0,\tfrac12)^{K}, (0,1), (1,1)
  \right] ]
    F_D^{(K)} (
      1; \underbrace{1,\ldots,1}_{K};
      K+1; -\frac{\Omega_{t,r}}{\beta_1},
              \ldots,
              -\frac{\Omega_{t,r}}{\beta_K} ),
\label{Asm_finnal}
\end{equation}
\hrulefill
\vspace{-8pt}
\end{figure*}
\end{theorem}

\begin{proof}
The proof is elaborated in Appendix \ref{APPEN_Corollary1}.
\end{proof}

\vspace{-0.23cm}
\subsection{Secrecy Diversity Order}
While the high-SNR asymptotic analysis reveals the general scaling behavior of the SOP, it is equally important to quantify the rate at which secrecy performance improves with increasing transmit power. This is captured by the SDO, which serves as a fundamental metric to characterize the asymptotic slope of the SOP curve and thus the robustness of AmbShield against eavesdropping. The SDO is derived in the following theorm.

\begin{theorem} \label{corollary2}
The SDO is given by \eqref{Asm_finnal}.

\vspace{-0.3cm}
\end{theorem}
\begin{equation}
  D_s  =
  - \lim_{P\to\infty}\frac{\log P_{\mathrm{SOP}}^{\mathrm{Asy}}}{\log P } =1.
\end{equation}

\begin{proof}
The proof is elaborated in Appendix \ref{APPEN_Corollary2}.
\end{proof}

\textit{Remark 2} 
From Theorems \ref{corollary1} and \ref{corollary2}, we can find that: 1) Increasing transmit power can improve the received SNR at the destination, while the SINR at
the eavesdropper gets close to constant, thereby boosting secrecy capacity and strengthening security; 2) When the number of BDs $K$ increases, the received SNR at the destination and the interference at the eavesdropper can be heightened due to the enhancement of the backscattered signal, and thus secrecy performance becomes enhanced; 3) The system can enjoy secrecy diversity gain owing to secrecy diversity order is 1.



\begin{table}[t]
\centering
\caption{Default System Parameters}
\label{tab:default_system_params}
\resizebox{\linewidth}{!}{
\begin{tabular}{l l}
\hline
\textbf{System parameters} & \textbf{Value} \\
\hline
Transmit Power ($P_t$)                         & 20 dBm \\
Number of BDs ($K$)                      &20              \\
The distance between transmitter and Eve (\(d_{t,e}\))   & \(10\) $\mathrm{m}$ \\
The distance between $k$-th BD and Eve (\(d_{k,e}\))   & \(2\) $\mathrm{m}$ \\
The distance between transmitter and $k$-th BD (\(d_{t,k}\))   & $\mathrm{9.5}$ $\mathrm{m}$\\
   The distance between transmitter and receiver   (\(d_{t,r}\))  & \(10\) $\mathrm{m}$ \\
        The activity range of $k$-th BD around the receiver (\(d_{k.r}\))& 0.5 $\mathrm{m}$ \\
Noise power received at Eve and receiver  (\(\sigma_e, \sigma_r\))  & -\( 30\, dBm\) \\ 
      Backscatter coefficient   ($\alpha_k$)& 0.8\\ 
      Path loss factor ($\alpha$)   & 2.5\\
      Normalized timing error $T_s$   & 0-0.3\\
CSI estimation error $\epsilon_{csi}$   & 0-0.4\\
    Channel fading    ($\tilde h_{ij}$)& Rayleigh fading\\ 
      Target secrecy rate threshold $R_S$   & 1 bps/Hz\\

\hline
\end{tabular}
}
\vspace{-0.5cm}
\end{table}

\section{Performance Evaluation} \label{performance}

In this section, we describe the experimental setup and evaluate the proposed AmbShield system regarding the metrics analyzed, i.e., CDF and SOP. Specifically, we first present the CDF to analyze the effectiveness of the proposed system. Then, we demonstrate the SOP of the proposed system under various practical parameters to reveal its potential in the future. 

\subsection{Simulation Setup}
We consider a downlink wireless communication system consisting of a transmitter, a receiver, an eavesdropper, and multiple AmBDs. Following the system model described in Section \ref{System model}, the transmitter emits a downlink signal intended for the receiver, while an eavesdropper simultaneously intercepts this transmission in an attempt to extract confidential information \cite{10337761}.
The AmBDs, distributed close to the receiver, randomly reflect the incident downlink signals. The legitimate receiver possesses prior knowledge of these reflection patterns, enabling effective decoding of reflected signals \cite{10109228}. 
In contrast, the eavesdropper lacks this knowledge, rendering it incapable of decoding the reflected signals \cite{10599122}.
The proposed system is also applicable to the uplink transmission with reserved roles of the transmitter and the receiver.
To validate the effectiveness of security and theoretical analyses presented in Section \ref{SNR Distribution}, Monte Carlo simulations (``\textbf{Simulation}''), theoretical curves (``\textbf{Theory}''), and asymptotic results (``\textbf{Asymptotic}''). are employed to evaluate the performance of the proposed system under various parameters. We adopt practical system parameters from existing studies \cite{10337761,10271138} in the related scenario, which are summarized in Table \ref{tab:default_system_params}.

\vspace{-0.5cm}
\subsection{Simulation Results}

\begin{figure}[!tb]
    \centering
    \subfigure[Legitimate Receiver]{
        \includegraphics[width=0.47\columnwidth]{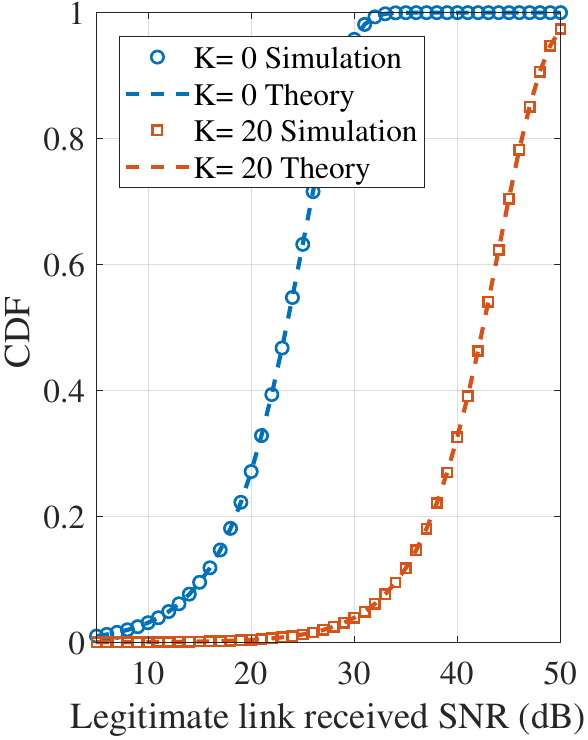}
        \label{CDF1}
    }
    \hspace{-6mm}
    \subfigure[Eve]{
        \includegraphics[width=0.47\columnwidth]{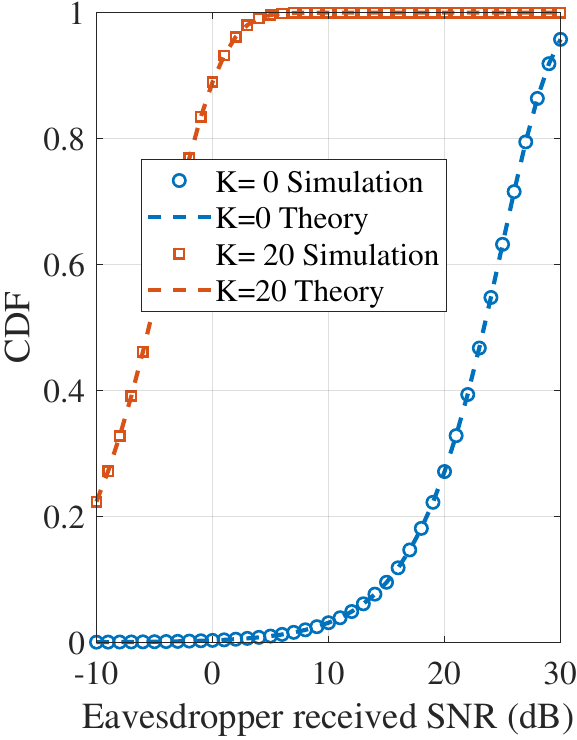}
        \label{CDF2}
    }
    \caption{The CDF of the received SNRs at the legitimate receiver and the Eve, respectively.
}
    \label{CDF}
\end{figure}

Figures \ref{CDF1} and \ref{CDF2} illustrate the CDF of the received SNR at the legitimate receiver and the Eve, respectively. The results compare scenarios without an AmBD (that is, $K$ = 0) and scenarios with multiple AmBDs.
It could be found that both the analytical CDF of the received SNRs at the legitimate user and Eve match well with the numerical curves.
Furthermore, in Figure \ref{CDF1}, we observe that the legitimate receiver experiences a significant improvement in the received SNR distribution when there are multiple AmBDs. Specifically, the CDF curves shift to the right, indicating that higher SNR values become more probable, thus demonstrating the positive impact of AmBDs in enhancing legitimate communication quality.
In contrast, Figure \ref{CDF2} reveals the opposite effect on the eavesdropper side. The presence of multiple AmBDs causes the received SNR at the eavesdropper to deteriorate, shifting the CDF curves to the left. This indicates that the eavesdropper experiences lower SNR values more frequently as a result of the interference created by the random reflections of the AmBDs, effectively degrading its interception capability.
Overall, the experimental results confirm that the incorporation of AmBDs effectively enhances communication security by simultaneously improving the quality of legitimate signals and impairing the eavesdropper's received signal. 

\begin{figure}[tp]
    \centering
    \includegraphics[width=0.9\linewidth]{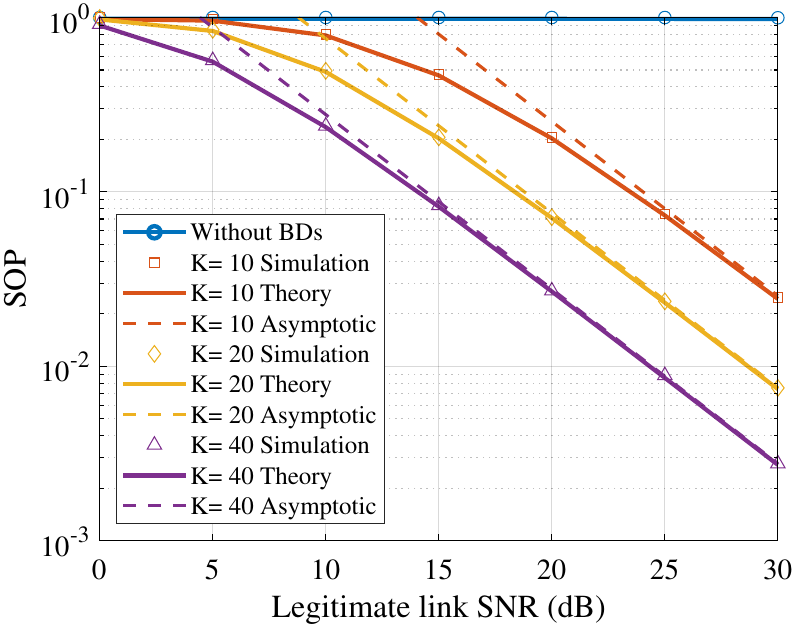}
    \caption{The SOP versus SNR of legitimate Channel under different BD numbers}
    \label{SOPSNR}
    \vspace{-0.5cm}
\end{figure}

Figure~\ref{SOPSNR} presents the relationship between the SOP and the legitimate link SNR, considering different numbers of BD. It is evident that the presence of AmBDs significantly enhances the system's security performance by lowering the SOP across all examined SNR values. Specifically, situations with a larger number of AmBDs exhibit progressively reduced SOP, highlighting that increased deployment of AmBDs provides improved resilience against eavesdropping. Furthermore, the strong alignment observed among the Monte Carlo simulations, theoretical curves, and asymptotic analyses underscores the validity and reliability of the proposed theoretical model. This analysis confirms that incorporating additional AmBD can effectively mitigate security risks, substantially improving the performance of the system's secrecy.

\begin{figure}[tp]
    \centering
    \includegraphics[width=0.9\linewidth]{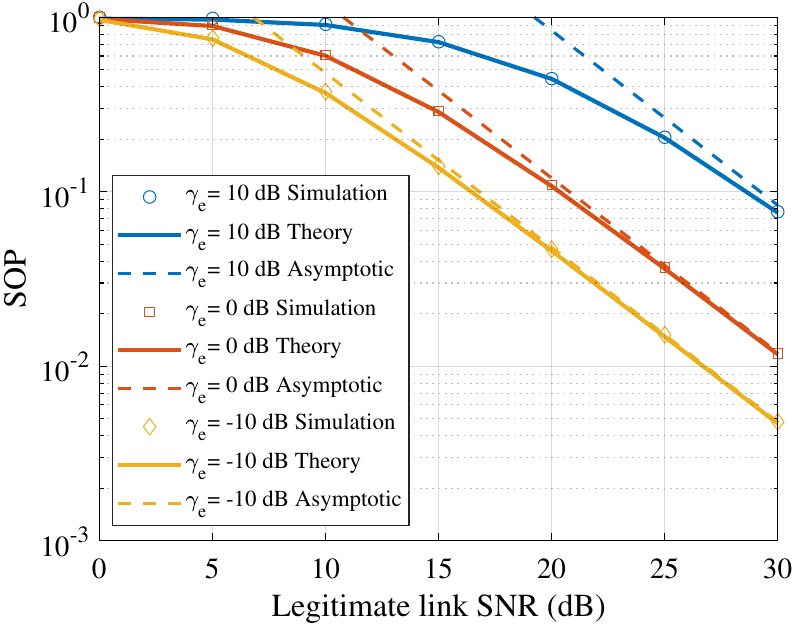}
    \caption{The SOP versus SNR of legitimate Channel under different $\bar\gamma_e$}
    \label{SOPEVESNR}
\end{figure}

Figure~\ref{SOPEVESNR} demonstrates how varying the average SNR in the eavesdropper (${\gamma}_e$) impacts the SOP. It is evident that a lower average SNR for eavesdroppers significantly improves secrecy performance. Specifically, scenarios where ${\gamma}_e$ is reduced (e.g., from 10 to -10 dB) exhibit consistently lower SOP values throughout the SNR range of the legitimate link. This trend highlights that a weaker eavesdropping channel substantially hinders the eavesdropper's ability to intercept confidential transmissions, thereby enhancing overall system security. Furthermore, the close alignment between the simulation, theoretical, and asymptotic results confirms the validity and precision of the analytical framework used.

\begin{figure}[tp]
    \centering
    \includegraphics[width=0.9\linewidth]{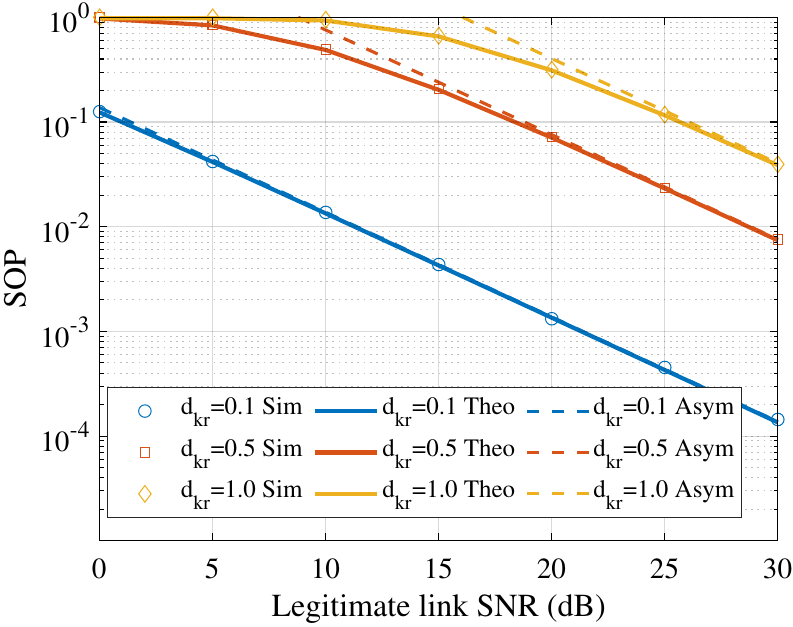}
    \caption{The SOP versus SNR of legitimate Channel under different $d_{k,r}$}
    \label{SOPDis}
    \vspace{-0.5cm}
\end{figure}

Figure~\ref{SOPDis} examines the impact of varying the distance between BD and the legitimate receiver ($d_{k,r}$) on the SOP. The results clearly demonstrate that decreasing the distance between the AmBDs and the legitimate receiver significantly enhances the performance of secrecy. Specifically, scenarios with $d_{k,r} = 0.1$ m show substantially lower SOP values compared to those with distances of 0.5 m and 1.0 m, across the entire SNR range of legitimate links. Since AmBDs are typically located close to the receiver (less than 0.5 m), this demonstrates the practicality of employing AmBDs to improve security and highlights the importance of positioning BDs for performance optimization.

\begin{figure}[tp]
    \centering
    \includegraphics[width=0.9\linewidth]{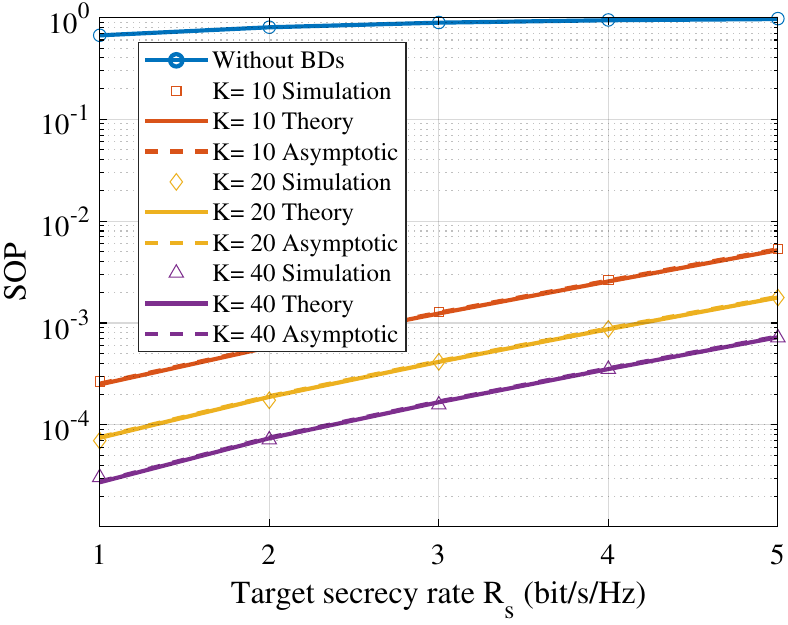}
    \caption{The SOP versus $R_s$ under different AmBDs}
    \label{SOPRs}
    \vspace{-0.4cm}
\end{figure}

Figure~\ref{SOPRs} evaluates the SOP as a function of varying target secrecy rates ($R_s$). As observed, increasing the target secrecy rate leads to higher SOP values due to the heightened security requirements. However, the presence of AmBDs significantly reduces this adverse impact. Specifically, scenarios with 10, 20, and 40 AmBDs consistently achieve lower SOP across all secrecy rate thresholds compared to scenarios without BD. This clearly illustrates the efficacy of deploying multiple AmBDs to mitigate performance degradation and ensure robust system security even under demanding secrecy conditions. The close alignment between simulation, theoretical, and asymptotic results further validates the accuracy and practical relevance of the theoretical analysis presented.

\begin{figure}[tp]
    \centering
    \includegraphics[width=0.9\linewidth]{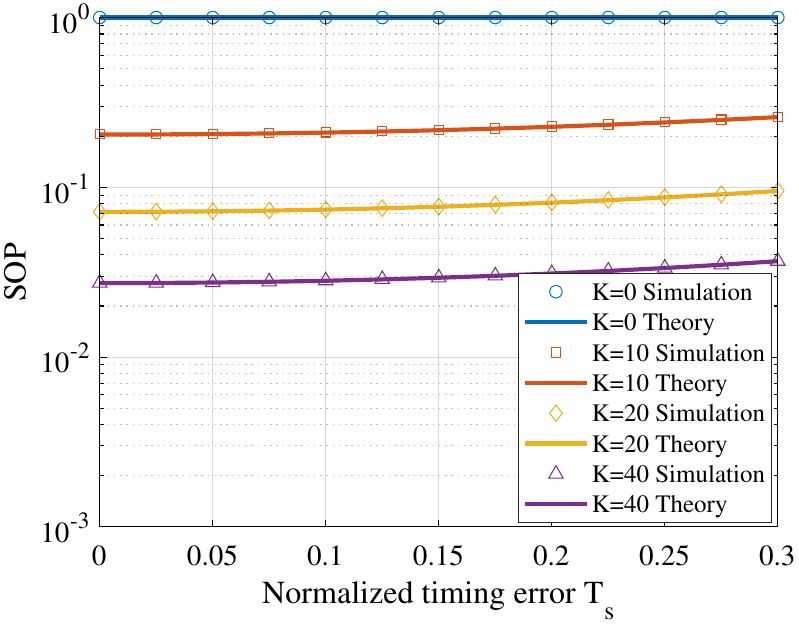}
    \caption{The SOP versus normalized timing error under different AmBDs when SNR = 20dB}
    \label{timeerror}
    \vspace{-0.5cm}
\end{figure}

\begin{figure}[tp]
    \centering
    \includegraphics[width=0.9\linewidth]{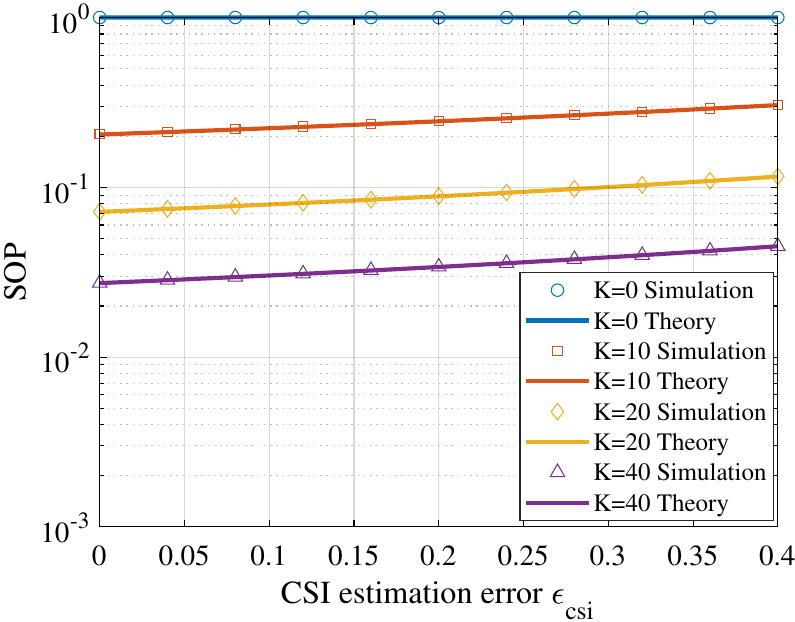}
    \caption{The SOP versus normalized mean squared estimation error under different AmBDs when SNR = 20dB}
    \label{CSIerror}
    \vspace{-0.5cm}
\end{figure}
Figure~\ref{timeerror}  investigates the SOP as the normalized time synchronization error between the AmBD and the legitimate receiver $T_s$ increases, for different numbers of AmBD. The curves show a clear monotonic increase in SOP with $T_s$, reflecting the degradation of accuracy in the legitimate receiver as the BD-assisted paths lose temporal alignment. In addition,
the effect of the synchronization error on the SOP is limited, which suggests that the modest synchronization error can be tolerated in AmbShield.

Figure~\ref{CSIerror} evaluates the SOP under increasing CSI estimation error of the legitimate channel. As the estimation error grows, the SOP increases because the effective SNR of the legitimate link is reduced by estimation inaccuracies, tightening the secrecy margin against the eavesdropper. Nevertheless, scenarios with more AmBDs consistently achieve lower SOP throughout the estimation error range, demonstrating that AmBD-assisted diversity and constructive relaying can compensate for imperfect CSI to a large extent. Together with Fig. 8, these results reinforce that AmbShield remains effective under practical impairments, synchronization, and CSI estimation errors.







\vspace{-0.3cm}
\section{Conclusion} \label{conclusion}
In this paper, we have proposed AmbShield, an AmBD-assisted PLS framework that exploits naturally distributed AmBDs to jointly enhance the legitimate channel and degrade eavesdroppers, without extra transmit power, costly hardware, and eavesdropper CSI. We have established a practical system model in which the receiver benefits from known BD reflections while the eavesdropper is impaired by random backscatter, and derived exact SINR distributions that led to closed-form SOP expressions and high-SNR asymptotics, revealing a secrecy diversity order of one.
Our analysis and simulations have consistently shown that secrecy improves by deploying more AmBDs, positioning them closer to the receiver, and reducing the average SNR of the eavesdropper. In practical settings, AmbShield achieves an SOP below $10^{-3}$, and the resilience of AmbShield to imperfect synchronization and CSI estimation is demonstrated.  These findings provide both theoretical insights and concrete design guidelines for integrating backscatter devices into future low-power secure wireless systems. 

\appendices

\section{Derivation of the PDF of \(V_k = |\phi_k h_{k,r} h_{t,k}|^2\)} \label{Appendx: V_k}

Let $X = |h_{k,r}|^2 \sim \mathrm{Exp}  (\tfrac{1}{\Omega_{k,r}} ), Y = |h_{t,k}|^2 \sim \mathrm{Exp}  (\tfrac{1}{\Omega_{t,k}} )$, where \(X\) and \(Y\) are independent.  Define the intermediate variable $Z = X Y$ and $\beta_k = \phi_k^2 \Omega_{k,r} \Omega_{t,k}$.

%


For two independent nonnegative random variables \(X\) and \(Y\), the PDF of \(Z=XY\) is
\begin{equation}
f_Z(z)
= \int_{0}^{\infty}
    f_X(x)
    f_Y\left(\tfrac{z}{x}\right)
    \frac{1}{x} \mathrm{d}x .
\end{equation}

Since $f_X(x) = \frac{1}{\Omega_{k,r}} e^{-x/\Omega_{k,r}},
f_Y (\tfrac{z}{x} )
= \frac{1}{\Omega_{t,k}}
  \exp  (-\frac{z}{x \Omega_{t,k}} )$, we obtain

\begin{equation}
f_Z(z)
= \frac{1}{\Omega_{k,r} \Omega_{t,k}}
  \int_{0}^{\infty}
    x^{-1}
    \exp \left(
      -\tfrac{x}{\Omega_{k,r}}
      -\tfrac{z}{\Omega_{t,k}x}
    \right)
   \mathrm{d}x .
\end{equation}

According to Macdonald function \cite{gradshteyn2014table}, we have $\int_{0}^{\infty}
  x^{-1} 
  \exp  (-a x - b/x )
 \mathrm{d}x
= 2 K_{0}  (2\sqrt{ab} )$.
Here set $a = \frac{1}{\Omega_{k,r}}, b = \frac{z}{\Omega_{t,k}}$, so that
\begin{equation}
f_Z(z)
= \frac{2}{\Omega_{k,r} \Omega_{t,k}}
   K_{0} \left(
     2\sqrt{\tfrac{z}{\Omega_{k,r} \Omega_{t,k}}}
   \right).
\end{equation}

Since \(V = \phi_k^2 Z\), a change of variables give $f_V(v)
= \frac{1}{\phi_k^2} 
  f_Z  (\frac{v}{\phi_k^2} )$. Therefore, the final PDF is

\begin{align}
f_V(v)
=  \frac{2}{\phi_k^{2}\Omega_{k,r} \Omega_{t,k}}
    K_{0} \left(
      2\sqrt{\tfrac{v}{\phi_k^{2}\Omega_{k,r} \Omega_{t,k}}}
    \right) = \frac{2}{\beta_k}
    K_{0} \left(
      2\sqrt{\tfrac{v}{\beta_k}}
    \right).
\end{align}

\section{Derivation of the Fox--$H$ PDF for  in \eqref{eq:sinrPDF}}
\label{APPEN_LPDF}

To obtain \eqref{eq:sinrPDF}, we first use the multi-Laplace extension for each zero-order Hankel transform  of BD cascade link $V_{k}$:

\begin{equation}
\frac{1}{1+\beta_k t}
=\int_{0}^{\infty} e^{-v_k} e^{-\beta_k t v_k} dv_k .
\label{id:prodLap}
\end{equation}

Define auxiliary variable $W=\lvert X\rvert^{2}$, according to Hankel-Bessel nuclear integral formula, the PDF of $W$ can be calculated as \cite{watson1922treatise}

\begin{equation}
f_{W}(w)=\int_{0}^{\infty}J_{0} (2\sqrt{wt} )
            \Phi_{X}(t) \mathrm dt.
\label{nuclearintegralformula}
\end{equation}

Substitute \eqref{id:prodLap} into \eqref{nuclearintegralformula}, $f_{W}(w)$ can be further calculated as

\begin{align}
f_W(w) 
= & \int_{0}^{\infty} e^{-\sum v_k}
   \Biggl[
     \int_{0}^{\infty}
       e^{-(\Omega_{t,r}+\sum \beta_k v_k)t} 
       J_{0} \left(2\sqrt{w t}\right) 
       \mathrm{d}t
   \Biggr] \nonumber \\[4pt]
&\times \mathrm{d}\mathbf{v}.
\label{eq:innerInt}
\end{align}

For the Laplace transform of the zero-order Bessel kernel, we have 
$\int_{0}^{\infty} e^{-c t}
  J_0  (2\sqrt{w t} ) dt
=\frac{1}{c}\exp  (-\frac{w}{c} ),$ where $c=\Omega_{tr}+\sum\beta_k v_k$. Therefore, \eqref{eq:innerInt} can be further transformed to

\begin{equation}
f_W(w) =
 \int_{0}^{\infty}
   \frac{e^{-\sum v_k}}
        {\Omega_{t,r} + \sum \beta_k v_k} 
   \exp \left(
     -\tfrac{w}{\Omega_{t,r} + \sum \beta_k v_k}
   \right)
  \mathrm{d}\mathbf{v}.
\label{eq:postLap}
\end{equation}

To obtain the Fox-H expression, \eqref{eq:postLap} needs to be measured and scaled uniformly. We define $z=\frac{w}{\Omega_{tr}}, u_k=\beta_k v_k,$ and $\Sigma=\sum_{k=1}^{K}\frac{u_k}{\Omega_{tr}}$. Thus, $d\mathbf v = (\prod \beta_k^{-1}) d\mathbf u$ and $\Omega_{tr}+\sum\beta_k v_k = \Omega_{tr}(1+\Sigma)$, and \eqref{eq:postLap} then becomes 

\begin{equation}
f_W(w) =
 \frac{1}{\Omega_{t,r}}
 \left(\prod_{k=1}^{K} \beta_k^{-1}\right)
 \int_{0}^{\infty}
   \frac{e^{-\sum u_k/\beta_k}}
        {1+\Sigma} 
   \exp \left(
     -\tfrac{z}{1+\Sigma}
   \right)
  \mathrm{d}\mathbf{u}.
\label{eq:dimless}
\end{equation}

By using the standard double-Gamma Mellin--Barnes kernel $\frac{1}{1+\Sigma} 
\exp  (-\frac{z}{1+\Sigma} )
=\frac{1}{2\pi i}
 \int_{\mathcal L}
   \Gamma(-s)\Gamma(1+s) 
   z^{ s}
   (1+\Sigma)^{-1-s}
 ds$, \eqref{eq:dimless} can be reformatted as a Mellin–Barnes contour integral, shown in \eqref{eq:swapInt}.

\begin{figure*}[t]
			\normalsize
\begin{equation}
f_W(w) =
 \frac{1}{\Omega_{t,r}}
 \left(\prod_{k=1}^{K} \beta_k^{-1}\right)
 \frac{1}{2\pi i}
 \int_{\mathcal{L}}
   \Gamma(-s) \Gamma(1+s) z^{s}
   \left[
     \underbrace{
       \int_{0}^{\infty}
         e^{-\sum u_k/\beta_k}
         (1+\Sigma)^{-1-s}
        \mathrm{d}\mathbf{u}
     }_{\text{Inner Integral}}
   \right]
 \mathrm{d}s .
\label{eq:swapInt}
\end{equation}
\hrulefill
\vspace{-8pt}
\end{figure*}

It can be found that the inner integral part in \eqref{eq:swapInt} can be evaluated via a Dirichlet change of variables


\begin{equation}
\begin{aligned}
&\int_{\mathbb{R}_+^{K}}
   \exp\!\left(-\sum_{k=1}^{K}\frac{u_k}{\beta_k}\right)
   \bigl(1+\Sigma\bigr)^{-1-s}\,
   \mathrm{d}\mathbf{u}  \\[6pt]
&\quad=\;
   \bigl[\Gamma(1+s)\bigr]^{K}\,
   \Omega_{t,r}^{\,(K-1)s}\,
   \left(\prod_{k=1}^{K}\beta_k\right)^{s},
\end{aligned}
\label{eq:innerRes}
\end{equation}

Substituting \eqref{eq:innerRes} into \eqref{eq:swapInt}, $f_W(w)$ can be further denoted as

\begin{equation}
f_W(w)=
\frac{1 }{\Omega_{tr}}
\frac{1}{2\pi i}
\int_{\mathcal L}
  \underbrace{\Gamma(-s)}_{\rm upper}\underbrace{[\Gamma(1+s)]^{K+1}}_{\rm bottom}
  \left(\frac{w \Omega_{tr}^{K-1}}{\prod\beta_k}\right)^{ s}
ds,
\label{eq:preH}
\end{equation}
which follows a Fox-H style.
Define the Fox--$H$ argument $\Xi  =  \frac{w \Omega_{tr}^{ K-1}}{\prod_{k=1}^{K}\beta_k}$ and according to \eqref{eq:FoxH_def}, it can be found that $p = 1, q = K+1, m= K+1$, and $n=1$. The Fox-H expression of \eqref{eq:preH} can be obtained as

\begin{equation}
f_W(w)=
\frac{1 }{\Omega_{tr}} 
H_{1,K+1}^{K+1,1} \left[
   \Xi
     | 
   \begin{array}{c}
     (1,1)\\[2pt]
     \underbrace{(1,1),\ldots,(1,1)}_{K+1}
   \end{array}
\right].
\label{eq:finalFoxH}
\end{equation}

For $\gamma_r = \dfrac{P}{\sigma_r^2}W$, the PDF of the legitimate channel can be calculated as:

\begin{align}
f_{\gamma_r}(\gamma)
&=\frac{\sigma_r^2}{P} f_W   (\tfrac{\sigma_r^2}{P}\gamma )\nonumber
 \\
&=\frac{\sigma_r^2 }{\Omega_{tr}P} 
H_{1,K+1}^{K+1,1}  \left[
\frac{\sigma_r^2\gamma \Omega_{tr}^{ K-1}}{P\prod_k\beta_k}
  | 
\begin{array}{c}
(1,1)\\[2pt]   \underbrace{(0,1),\ldots,(0,1)}_{K+1}
\end{array}
\right].
\end{align}
\section{Derivation of the PDF for \eqref{eq: EPDF}} \label{APPEN_EPDF}

For compactness, we introduce intermediate variable $\lambda    = \frac{1}{P\Omega_{t,e}}, 
A_k =    P a_k^{2}\Omega_{t,k}\Omega_{k,e},$ and 
$u_k(x) = \frac{P\Omega_{t,e}}{x A_k}$.

Now, the CDF of $\gamma_e$ in \eqref{eq:CDF_generic} becomes

\begin{equation}
F_{\gamma_e}(x)
   =1-e^{-\lambda\sigma_e^{2}x} 
      \mathcal L_I(\lambda x),  
\mathcal L_I(s)=\prod_{k=1}^{K}\mathcal L_{W_k}(sA_k).
\end{equation}

Differentiating the CDF, the PDF of $\gamma_e$ yields

\begin{align}
f_{\gamma_e}(x)
 &=\frac{d}{dx} [1-F_{\gamma_e}(x) ]\notag\\
 &=\lambda\sigma_e^{2} e^{-\lambda\sigma_e^{2}x}
     \mathcal L_I(\lambda x)
   -\lambda e^{-\lambda\sigma_e^{2}x} 
    \mathcal L_I'(\lambda x).
\label{S1}
\end{align}

To obtain the derivative of $\mathcal L_I$, we have 
\begin{equation}
\frac{d}{ds}\ln\mathcal L_{W_k}(s)
   =-\frac{1+u_k-\exp(-u_k)/E_1(u_k)}{s},
\end{equation}

so for $\mathcal L_I'(\lambda x)$, we have  

\begin{equation} 
\mathcal L_I'(\lambda x)=
     \mathcal L_I(\lambda x)
      \frac{P\Omega_{t,e}}{x}
     \sum_{k=1}^{K} 
        [
         -1-u_k(x)+\tfrac{e^{-u_k(x)}}{E_1(u_k(x))}
        ].
\label{S2}
\end{equation}

Substituting \eqref{S2} into \eqref{S1}, the final PDF of $\gamma_e$  can be obtained as follows: 

\begin{align}
f_{\gamma_e}(x) = &
 \lambda  
 e^{-\lambda \sigma_{e}^{2} x} 
 \mathcal{L}_I(\lambda x)
\nonumber \\[2pt]
&\times 
 \left[
   \sigma_{e}^{2}
   + \frac{P \Omega_{t,e}}{x}
     \sum_{k=1}^{K} 
       \left(
         1 + u_k(x)
         - \frac{e^{-u_k(x)}}{E_1 \left(u_k(x)\right)}
       \right)
 \right].
\end{align}

\section{Derivation of the Fox-H expression PDF for the SOP result \eqref{eq:SOP_final_corrected}} \label{APPEN_SOP}

Using the CDF of \(\gamma_r\) and the PDF of \(\gamma_e\), \eqref{eq:SOP_def} becomes the single integral
\begin{equation}\label{eq:SOP_master}
P_{\rm SOP}
  =\int_{0}^{\infty}
     F_{\gamma_r}   \bigl(\Theta(1+x)-1 )  
     f_{\gamma_e}(x)  \mathrm dx.
\end{equation}

Substituting \eqref{eq:grmmar} and \eqref{eq: EPDF} into \eqref{eq:SOP_master} yileds

\begin{equation}
P_{\rm SOP}
=\frac{1}{\Omega_{t,r}}
\int_{0}^{\infty}
      g(x) 
      \underbrace{H_{2,K+2}^{K+1,2}   \bigl(A_rg(x) )}_{%
      \substack{\text{legitimate‑link}\\[0.1em]\text{Fox–$H$ CDF}}} 
      \underbrace{f_{\gamma_e}(x)}_{\substack{\text{PDF of}\\[0.1em]\text{wiretap SINR}}}
       \mathrm dx,
\label{eq:SOP_integral}
\end{equation}
where $g(x)=   \Theta x+\Theta-1$, and  $A_r=\dfrac{\sigma_r^{2}\Omega_{t,r}^{K-1}}{P\prod_{k=1}^{K}\beta_k}$ represent the scale factor in legitimate CDF.
Using the Mellin definition, we expand the legitimate-link
Fox–H CDF term in \eqref{eq:SOP_integral} to 

\begin{equation}
H_{2,K+2}^{K+1,2}(A_rg)
=\frac{1}{2\pi i}\int_{\cal C_s}
   \Phi(s)  (A_rg)^{-s}   \mathrm ds,
   \label{eq:H}
\end{equation}
where $\Phi(s)=
\Gamma(1-s)\vphantom{\bigl|^{|^|}}
\bigl[\Gamma(s) ]^{K+1}.$


Insert \eqref{eq:H} into \eqref{eq:SOP_integral} yields:
\begin{equation}
P_{\rm SOP}
=
\frac{1}{\Omega_{t,r}}
\frac{1}{2\pi i}
\int_{\mathcal{L}_s}
\Phi(s)(A_r)^{-s}
\underbrace{\int_{0}^{\infty}g(x)^{1-s}f_{\gamma_e}(x) \mathrm{d}x}_{\mathcal{I}(s)}
 \mathrm{d}s.
 \label{SOP33}
\end{equation}

Define $g(x)=\Theta x + \Theta - 1 = (\Theta-1)  (1 + d x),
c= \Theta - 1, d = \frac{\Theta}{c}$, the inner $x$–integral $\mathcal I(s)$ in \eqref{SOP33} can be rewritten as an explicit dependence on the threshold $\Theta$.

\begin{equation}
\mathcal{I}(s)=c^{1-s}\int_{0}^{\infty}(1 + d  x)^{1-s}  f_{\gamma_e}(x)  \mathrm{d}x. 
\label{Is}
\end{equation}

Now $g(x)^{1-s}=c^{1-s}(1+dx)^{1-s}$, let $\mu=1-s$. Expand
\((1+dx)^{u} \) via the MB identity

\begin{equation}
(1 + d x)^\mu
= \frac{1}{2\pi i}
  \int_{C_u}
    \frac{\Gamma(-u) \Gamma(\mu + u)}
         {\Gamma(\mu)}
     (d x)^u
   \mathrm{d}u.
   \label{1+dx}
\end{equation}

Meanwhile, the eavesdropping PDF \eqref{eq: EPDF} can be compacted to

\begin{equation}
f_{\gamma_e}(x)
=\lambda 
e^{-\lambda\sigma_e^{2}x} 
H_{2K+2,0}^{0,2K+2} \bigl[\lambda x  \bigm| \mathcal A ],
\label{eq:PDF_eve_compact}
\end{equation}
where $\mathcal A
=\bigl[(0,1)^K, (0,\tfrac12)^K, (0,1), (1,1) ]$

Substituting \eqref{1+dx} and \eqref{eq:PDF_eve_compact} into \eqref{Is}, and interchanging the $u$ and $x$ integrals, we have \eqref{Is2}.

\begin{figure*}[t]
			\normalsize
\begin{equation}
\mathcal{I}(s)
= c^{1-s}  
\frac{\lambda}{2\pi i  \Gamma(s-1)}
\int_{C_u}
\Gamma(-u)  \Gamma(s-1+u)  d^{  u}  
\underbrace{%
  \int_{0}^{\infty}
    x^{-u}  e^{-\lambda\sigma_e^2 x}  
    H_{2K+2,0}^{0,2K+2}\bigl[\lambda x ]
    \mathrm{d}x
}_{\mathcal{L}(u)}
  \mathrm{d}u.
\label{Is2}
\end{equation}
\hrulefill
\vspace{-8pt}
\end{figure*}

In \eqref{Is2}, we can find that the $\Theta$-correlated scale $c^{1-s}$ have been isolated, and $\mathcal{L}(u)$ is a functions on channel statistics, which can be calculated as 

\begin{align}
\mathcal L(u)
&=\int_{0}^{\infty}
      x^{-u} 
      e^{-\lambda\sigma_e^{2}x} 
      H_{2K+2,0}^{0,2K+2} \bigl[\lambda x \bigm|\mathcal A ]
      \mathrm{d}x   \nonumber \\
&=(\lambda\sigma_e^{2})^{-1+u} 
H_{2K+3,0}^{0,2K+3}
   \Bigl[\sigma_e^{-2}  \bigm| \mathcal A,(1-u,1) ].
   \label{LuFinal}
\end{align}

\begin{figure*}[t]
			\normalsize
\begin{equation}
\mathcal I(s)=
c^{1-s}  (\lambda\sigma_e^{2})^{-1}  
\frac{\lambda}{2\pi i  \Gamma(s-1)}
\underbrace{\int_{C_u} 
  \Gamma(-u)  \Gamma(s-1+u)  
  (\lambda\sigma_e^{2}d)^{u}  
  H_{2K+3,0}^{0,2K+3}
     \Bigl[\sigma_e^{-2}  \bigm|  
            \mathcal A,  (1-u,1) ]
\mathrm du}_{J(u)} .
\label{Is3}
\end{equation}
\hrulefill
\vspace{-8pt}
\end{figure*}

Substitute \eqref{LuFinal} into \eqref{Is2}, we have \eqref{Is3}, where the integral for $u$ in $J(u)$ can be calculated using Kilbas-Saigo Convolution Theorem \cite{mathai2009h} as follows:

\begin{align}\label{eq:A4}
\int_{C_u}
  &\Gamma(-u)  \Gamma(s-1+u)  
  x^{u}  
  H_{2K+3,0}^{0,2K+3}
   \Bigl[\sigma_e^{-2}  \bigm|  
          \mathcal A,  (1-u,1) ]
  \mathrm du \nonumber \\
&=
H_{2K+4,3}^{0,2K+3}
  \Bigl[\sigma_e^{-2}  \bigm|  
         \mathcal A,  (2-s,1) ]  
 x^{  s-1},
\end{align}

By substituting \eqref{eq:A4} and \eqref{Is3} into \eqref{SOP33}, the SOP can be obtained in \eqref{SOP_final}.

\begin{figure*}[t]
			\normalsize
\begin{align}
P_{\mathrm{SOP}}
&=\frac{1}{\Omega_{t,r}} 
 \frac{1}{2\pi i}
 \int_{C_s}
   \Phi(s) (A_r)^{-s} c^{1-s} d^{ s-1} \lambda 
   \bigl(\lambda\sigma_e^2 )^{-1} 
   H^{0,2K+3}_{2K+4,3}
   \Bigl[
     \sigma_e^{-2} \Bigm| \mathcal{A}, (2-s,1)
    ]
  \mathrm{d}s \nonumber \\
&= \frac{c}{\Omega_{t,r} d  \sigma_e^2} 
 \frac{1}{2\pi i}
 \int_{C_s}
   \Phi(s) 
   \Bigl(\frac{A_r c}{d} )^{-s} 
   H^{0,2K+3}_{2K+4,3}
   \Bigl[
     \sigma_e^{-2} \Bigm| \mathcal{A}, (2-s,1)
    ]
  \mathrm{d}s.
\label{SOP_final}
\end{align}
\hrulefill
\vspace{-8pt}
\end{figure*}

Applying the Kilbas-Saigo Convolution Theorem a second time eliminates the $s$‑integral, uplifting the Fox–$H$ orders by one in each row for the final SOP expression:

\begin{equation}
P_{\mathrm{SOP}}
= \frac{c}{\Omega_{t,r}  \sigma_{e}^{2}}
  H_{2K+3,0}^{0, 2K+3} \left[
    \tfrac{d}{A_r c} \sigma_{e}^{-2}
     \bigm| 
    (0,1), \mathcal{A}
  \right].
\end{equation}


\section{Proof of Theorem \ref{corollary1}} \label{APPEN_Corollary1}

Define the instantaneous threshold $T    =   \Theta  (1+\gamma_e )-1$, and the normalised transmit SNR $\rho    =   \frac{P}{\sigma_r^{2}}$. Thus, \eqref{eq:SOP_def} becomes 

\begin{equation}
P_{\mathrm{SOP}}
 = \Pr \{\gamma_r < \Theta (1 + \gamma_e) -1 \}
 = 
\int_{0}^{\infty}
F_{\gamma_r} (T(x) ) 
f_{\gamma_e}(x) 
\mathrm{d}x,
\label{eq: AsmSOP}
\end{equation}

Since $\gamma_r = \rho W$, $F_{\gamma_r}(y)$ can be further expressed as 

\begin{equation}
F_{\gamma_r}(y)
 = 
\Pr\{\rho W < y\}
 = 
F_{W}  (\tfrac{y}{\rho} ).
\label{F_yr}
\end{equation}

When $\rho\to\infty$, we have $\tfrac{y}{\rho} \to 0$. Thus, $F_{W}(w)$ can be expressed through a first-order Taylor expansion:

\begin{align}
F_{W}(w)
&=\left.\frac{\mathrm{d}F_{W}}{\mathrm{d}w}\right|_{w=0}w + 
o (\rho^{-1} )\nonumber \\
&=f_{W}(0) w + 
o (\rho^{-1} )\nonumber  \\
&=w\int_{0}^{\infty}
e^{-\Omega_{t,r} t}
\prod_{k=1}^{K}\frac{1}{1+\beta_{k}t}
 \mathrm{d}t + 
o (\rho^{-1} ).
\label{F_w}
\end{align}

Substituting \eqref{F_w} into \eqref{F_yr}, we have 

\begin{equation}
F_{\gamma_r}(T)
=
F_{W}  (\tfrac{T}{\rho} )
=
\frac{T}{\rho}\int_{0}^{\infty}
e^{-\Omega_{t,r} t}
\prod_{k=1}^{K}\frac{1}{1+\beta_{k}t}
 \mathrm{d}t + 
o (\rho^{-1} )
.
\label{eq:AsmFr}
\end{equation}

Substituting \eqref{eq:AsmFr} into \eqref{eq: AsmSOP}, we have

\begin{equation}
P^{\rm Asy}_{\mathrm{SOP}}
=
\frac{1}{\rho}\int_{0}^{\infty}
e^{-\Omega_{t,r} t}
\prod_{k=1}^{K}\frac{1}{1+\beta_{k}t}
 \mathrm{d}t
\int_{0}^{\infty}
T(x) f_{\gamma_e}(x) \mathrm{d}x,
\label{asy_2}
\end{equation}
where we can find that $\int_{0}^{\infty}
T(x) f_{\gamma_e}(x) \mathrm{d}x$ is a weighted average of the PDFs at the eavesdropping end. It can be calculated as

\begin{equation}
\int_{0}^{\infty}T(x) f_{\gamma_e}(x) \mathrm{d}x
=\mathbb{E}[T] =
\Theta (1 + \bar\gamma_e )-1,
\label{1}
\end{equation}
where \(\bar\gamma_e\) represents the average SINR at the eavesdropping link, which can be obtained from the Laplace transform of the interference at the eavesdropping end by a strict integral equation

\begin{equation}
\bar\gamma_e
= P \Omega_{t,e}
\int_{0}^{\infty}
e^{-\sigma_e^2 u}
\mathcal{L}_I(u) \mathrm{d}u,
\label{2}
\end{equation}
where $\mathcal{L}_I(u)$ has been calculated in \eqref{eq:LI_final}

By substituting \eqref{1} and \eqref{2} into \eqref{asy_2}, we obtain the asymptotic value as

\begin{equation}
\label{eq:SOP-1st}
P^{\rm Asy}_{\mathrm{SOP}}=
\frac{\sigma_r^{2}}{P \Omega_{t,r}} 
[\Theta (1+\bar\gamma_e)-1]
\underbrace{\int_{0}^{\infty}
           e^{-\Omega_{t,r}t}
           \prod_{k=1}^{K} \frac{dt}{1+\beta_k t}}
_{C_{\text{leg}}},
\end{equation}
where $(1+\beta_k t)^{-1} =\int_{0}^{\infty} e^{-u_k}e^{-\beta_k t u_k} du_k$, using Laplace transform. After inserting all the \(K\) factors, we obtain the following:

\begin{align}
C_{\text{leg}}
&=\int_{0}^{\infty} e^{-\Omega_{t,r}t}
  \prod_{k=1}^{K} (1+\beta_k t )^{-1} dt  \notag\\
&=\int_{0}^{\infty}  \cdots \int_{0}^{\infty}
   \frac{e^{-\sum_{k=1}^{K}u_k}}
        {\Omega_{t,r}+\sum_{k=1}^{K}\beta_k u_k}
    du_1 \cdots du_K .                       \label{eq:C_leg_multi_int}
\end{align}

This \(K\)-dimensional integral converges for any positive real \(\beta_k\). 
Use  Mellin-Barnes
\((\Omega_{t,r}+\sum\beta_k u_k)^{-1}
  =(2\pi i)^{-1} \int_{\cal C_s} 
    \Gamma(s)\Gamma(1-s) 
     (\Omega_{t,r} )^{-s}
     (\sum\beta_k u_k )^{ s-1} ds\), \eqref{eq:C_leg_multi_int} can be transformated to a standard  Lauricella-type functions  \cite{dickey1983multiple} as follows:

\begin{equation}\label{cleg}
C_{\text{leg}}
  = \frac{1}{\Omega_{t,r}} 
    F_{D}^{(K)} \left(
      1; 
      \underbrace{1, \ldots, 1}_{K};
      K+1; 
      -\tfrac{\Omega_{t,r}}{\beta_{1}}, 
      \ldots, 
      -\tfrac{\Omega_{t,r}}{\beta_{K}}
    \right).
\end{equation}


Applying
\(e^{-\sigma_e^{2}u}
 =(2\pi i)^{-1} \int_{\mathcal C_s} 
   \Gamma(s)(\sigma_e^{2}u)^{-s}ds\)
and the Mellin form of \(\mathcal L_I(u)\),  \(\gamma_e\) becomes  
\begin{equation}
\bar\gamma_e=
P\Omega_{t,e} 
H^{0, 2K+2}_{ 2K+2, 0} 
\left[
  \sigma_e^{-2}
    | 
  (0,1)^{K}, (0,\tfrac12)^{K}, (0,1), (1,1)
\right].
\label{gamme}
\end{equation}

Substituting \eqref{gamme} and \eqref{cleg} into \eqref{eq:SOP-1st}, the asymptotic expression of the SOP yields in \eqref{Asm_finnal}.

\begin{figure*}[t]
			\normalsize
            {\small
\begin{equation}
P^{\text{Asy}}_{\mathrm{SOP}}
= \frac{\sigma_{r}^{2}}{P \Omega_{t,r}}
  \Biggl[
    \Theta - 1
    + \Theta P \Omega_{t,e} 
      H^{0, 2K+2}_{2K+2, 0} \left[
        \sigma_{e}^{-2} \Bigm| 
        (0,1)^{K}, (0,\tfrac{1}{2})^{K}, (0,1), (1,1)
      \right]
  \Biggr]
  F_{D}^{(K)} \left(
    1; 
    \underbrace{1,\ldots,1}_{K};
    K+1; 
    -\tfrac{\Omega_{t,r}}{\beta_{1}}, 
    \ldots, 
    -\tfrac{\Omega_{t,r}}{\beta_{K}}
  \right).
\end{equation} }
\hrulefill
\vspace{-8pt}
\end{figure*}

\section{Proof of Theorem \ref{corollary2}} \label{APPEN_Corollary2}

The SDO is determined as follows:

\begin{equation}
  D_s  =
  - \lim_{P\to\infty}\frac{\log P_{\mathrm{SOP}}^{\mathrm{Asy}}}{\log P } .
\end{equation}
where $P_{\mathrm{SOP}}^{\mathrm{Asy}}$ is expressed in ~\eqref{eq:SOP-1st},

The instantaneous SINR at Eve is
\(
  \gamma_e
  =\dfrac{P|h_{t,e}|^{2}}
         {\sigma_e^{2}+P \displaystyle\sum_{k}a_k^{2}
                      |h_{t,k}|^{2}|h_{k,e}|^{2}}
  \).
Re-scaling both numerator and denominator by~\(P\) and letting
\(P\to\infty\) yields
\begin{equation}
  \lim_{P \to \infty} \gamma_e
   =
  \frac{|h_{t,e}|^{2}}{\sum_{k} a_k^{2}|h_{t,k}|^{2}|h_{k,e}|^{2}}
  ,
\end{equation}
where expectation $\bar{\gamma}_e^{(\infty)}
   =\mathbb{E}  [\gamma_e^{(\infty)} ]$
is a finite constant because the denominator is a
strictly positive continuous random variable. Defining the auxiliary variable $c_0 = \frac{\sigma_r^2}{\Omega_{t,r}}  C_{\rm leg}   [\Theta \left(1+\bar{\gamma}_e^{(\infty)}\right)-1 ]$ and $P_{\mathrm{SOP}}^{\mathrm{Asy}}
    =\frac{c_0}{P}$, we have
    
\begin{equation}
  D_s
  =-\lim_{P\to\infty}
     \frac{-\log P + \log c_0 + o(1)}{\log P}
  = 1.
\end{equation}


\bibliographystyle{IEEEtran}
\bibliography{IEEEabrv,reference}

\begin{thebibliography}{10}
\providecommand{\url}[1]{#1}
\csname url@samestyle\endcsname
\providecommand{\newblock}{\relax}
\providecommand{\bibinfo}[2]{#2}
\providecommand{\BIBentrySTDinterwordspacing}{\spaceskip=0pt\relax}
\providecommand{\BIBentryALTinterwordstretchfactor}{4}
\providecommand{\BIBentryALTinterwordspacing}{\spaceskip=\fontdimen2\font plus
\BIBentryALTinterwordstretchfactor\fontdimen3\font minus \fontdimen4\font\relax}
\providecommand{\BIBforeignlanguage}[2]{{%
\expandafter\ifx\csname l@#1\endcsname\relax
\typeout{** WARNING: IEEEtran.bst: No hyphenation pattern has been}%
\typeout{** loaded for the language `#1'. Using the pattern for}%
\typeout{** the default language instead.}%
\else
\language=\csname l@#1\endcsname
\fi
#2}}
\providecommand{\BIBdecl}{\relax}
\BIBdecl

\bibitem{10044183}
B.~Mao, J.~Liu, Y.~Wu, and N.~Kato, ``Security and privacy on {6G} network edge: A survey,'' \emph{IEEE Communications Surveys \& Tutorials}, vol.~25, no.~2, pp. 1095--1127, Feb. 2023.

\bibitem{zhang2025artificial}
Y.~Zhang, Z.~Yan, J.~Wang, Y.~Yang, Z.~Li, and R.~J{\"a}ntti, ``Artificial noise management for securing backscatter communication networks: Current advances, open challenges, and future directions,'' \emph{IEEE Network}, Early access, May 2025.

\bibitem{9524814}
V.~L. Nguyen, P.~C. Lin, B.~C. Cheng, R.~H. Hwang, and Y.-D. Lin, ``Security and privacy for {6G}: A survey on prospective technologies and challenges,'' \emph{IEEE Communications Surveys \& Tutorials}, vol.~23, no.~4, pp. 2384--2428, Aug.2021.

\bibitem{li2024ai}
J.~Li, Y.~Yang, and Z.~Yan, ``{AI}-based physical layer key generation in wireless communications: Current advances, open challenges, and future directions,'' \emph{IEEE Wireless Communications}, vol.~31, no.~5, pp. 182--191, Jul. 2024.

\bibitem{10623395}
C.~Zhao, H.~Du, D.~Niyato, J.~Kang, Z.~Xiong, D.~I. Kim, X.~Shen, and K.~B. Letaief, ``Generative {AI} for secure physical layer communications: A survey,'' \emph{IEEE Transactions on Cognitive Communications and Networking}, vol.~11, no.~1, pp. 3--26, Aug. 2025.

\bibitem{9350170}
C.~Gong, X.~Yue, Z.~Zhang, X.~Wang, and X.~Dai, ``Enhancing physical layer security with artificial noise in large-scale {NOMA} networks,'' \emph{IEEE Transactions on Vehicular Technology}, vol.~70, no.~3, pp. 2349--2361, Feb. 2021.

\bibitem{9349152}
M.~T. Mamaghani and Y.~Hong, ``Joint trajectory and power allocation design for secure artificial noise aided {UAV} communications,'' \emph{IEEE Transactions on Vehicular Technology}, vol.~70, no.~3, pp. 2850--2855, Feb. 2021.

\bibitem{10227884}
N.~Su, F.~Liu, and C.~Masouros, ``Sensing-assisted eavesdropper estimation: An {ISAC} breakthrough in physical layer security,'' \emph{IEEE Transactions on Wireless Communications}, vol.~23, no.~4, pp. 3162--3174, Aug. 2024.

\bibitem{10012331}
J.~Li, G.~Sun, H.~Kang, A.~Wang, S.~Liang, Y.~Liu, and Y.~Zhang, ``Multi-objective optimization approaches for physical layer secure communications based on collaborative beamforming in {UAV} networks,'' \emph{IEEE/ACM Transactions on Networking}, vol.~31, no.~4, pp. 1902--1917, Jan. 2023.

\bibitem{9439833}
J.~Zhang, H.~Du, Q.~Sun, B.~Ai, and D.~W.~K. Ng, ``Physical layer security enhancement with reconfigurable intelligent surface-aided networks,'' \emph{IEEE Transactions on Information Forensics and Security}, vol.~16, pp. 3480--3495, May 2021.

\bibitem{9763499}
Y.~Sun, K.~An, Y.~Zhu, G.~Zheng, K.-K. Wong, S.~Chatzinotas, D.~W.~K. Ng, and D.~Guan, ``Energy-efficient hybrid beamforming for multilayer {RIS}-assisted secure integrated terrestrial-aerial networks,'' \emph{IEEE Transactions on Communications}, vol.~70, no.~6, pp. 4189--4210, Apr. 2022.

\bibitem{10534116}
Y.~Zhang, S.~Zhao, Y.~Shen, X.~Jiang, and N.~Shiratori, ``Enhancing the physical layer security of two-way relay systems with {RIS} and beamforming,'' \emph{IEEE Transactions on Information Forensics and Security}, vol.~19, pp. 5696--5711, May 2024.

\bibitem{8737501}
P.~Wang, N.~Wang, M.~Dabaghchian, K.~Zeng, and Z.~Yan, ``Optimal resource allocation for secure multi-user wireless powered backscatter communication with artificial noise,'' in \emph{IEEE Conference on Computer Communications (INFOCOM)}, Paris, France, Apr. 2019, pp. 460--468.

\bibitem{9745773}
P.~Wang, Z.~Yan, N.~Wang, and K.~Zeng, ``Resource allocation optimization for secure multidevice wirelessly powered backscatter communication with artificial noise,'' \emph{IEEE Transactions on Wireless Communications}, vol.~21, no.~9, pp. 7794--7809, Mar. 2022.

\bibitem{8543573}
D.~Wang, B.~Bai, W.~Zhao, and Z.~Han, ``A survey of optimization approaches for wireless physical layer security,'' \emph{IEEE Communications Surveys \& Tutorials}, vol.~21, no.~2, pp. 1878--1911, Nov. 2019.

\bibitem{10130082}
T.~Jiang, Y.~Zhang, W.~Ma, M.~Peng, Y.~Peng, M.~Feng, and G.~Liu, ``Backscatter communication meets practical battery-free internet of things: A survey and outlook,'' \emph{IEEE Communications Surveys \& Tutorials}, vol.~25, no.~3, pp. 2021--2051, May. 2023.

\bibitem{3gppTR22840}
\BIBentryALTinterwordspacing
{3GPP}, ``{Study on Ambient IoT (AIoT)},'' {3rd Generation Partnership Project (3GPP)}, Tech. Rep. TR 22.840, Sep. 2023. [Online]. Available: \url{https://www.3gpp.org/DynaReport/22840.htm}
\BIBentrySTDinterwordspacing

\bibitem{yang2017modulation}
G.~Yang, Y.~C. Liang, R.~Zhang, and Y.~Pei, ``Modulation in the air: Backscatter communication over ambient {OFDM} carrier,'' \emph{IEEE IEEE Trans. Commun.}, vol.~66, no.~3, pp. 1219--1233, Nov. 2017.

\bibitem{8368232}
N.~Van~Huynh, D.~T. Hoang, X.~Lu, D.~Niyato, P.~Wang, and D.~I. Kim, ``Ambient backscatter communications: A contemporary survey,'' \emph{IEEE Communications Surveys \& Tutorials}, vol.~20, no.~4, pp. 2889--2922, May 2018.

\bibitem{abdulghafor2021recent}
R.~Abdulghafor, S.~Turaev, H.~Almohamedh, R.~Alabdan, B.~Almutairi, A.~Almutairi, and S.~Almotairi, ``Recent advances in passive {UHF-RFID} tag antenna design for improved read range in product packaging applications: A comprehensive review,'' \emph{IEEE Access}, vol.~9, Apr. 2021.

\bibitem{ahmed2024noma}
M.~Ahmed, M.~Shahwar, F.~Khan, W.~Ullah~Khan, A.~Ihsan, U.~Sadiq~Khan, F.~Xu, and S.~Chatzinotas, ``{NOMA}-based backscatter communications: Fundamentals, applications, and advancements,'' \emph{IEEE IoTJ}, vol.~11, no.~11, pp. 19\,303--19\,327, Apr. 2024.

\bibitem{10337761}
M.~Kaveh, Z.~Yan, and R.~Jäntti, ``Secrecy performance analysis of {RIS}-aided smart grid communications,'' \emph{IEEE Transactions on Industrial Informatics}, vol.~20, no.~4, pp. 5415--5427, Dec. 2024.

\bibitem{10109228}
X.~Li, J.~Jiang, H.~Wang, C.~Han, G.~Chen, J.~Du, C.~Hu, and S.~Mumtaz, ``Physical layer security for wireless-powered ambient backscatter cooperative communication networks,'' \emph{IEEE Transactions on Cognitive Communications and Networking}, vol.~9, no.~4, pp. 927--939, Apr. 2023.

\bibitem{poularikas2018transforms}
A.~D. Poularikas and A.~M. Grigoryan, \emph{Transforms and applications handbook}.\hskip 1em plus 0.5em minus 0.4em\relax CRC press, Sep. 2018.

\bibitem{coelho2019finite}
C.~A. Coelho and B.~C. Arnold, \emph{Finite form representations for Meijer G and fox H functions}.\hskip 1em plus 0.5em minus 0.4em\relax Springer, 2019.

\bibitem{mathai2009h}
A.~M. Mathai, R.~K. Saxena, and H.~J. Haubold, \emph{The H-function: theory and applications}.\hskip 1em plus 0.5em minus 0.4em\relax Springer Science \& Business Media, 2009.

\bibitem{9319204}
X.~Li, M.~Zhao, M.~Zeng, S.~Mumtaz, V.~G. Menon, Z.~Ding, and O.~A. Dobre, ``Hardware impaired ambient backscatter {NOMA} systems: Reliability and security,'' \emph{IEEE Transactions on Communications}, vol.~69, no.~4, pp. 2723--2736, Jan. 2021.

\bibitem{oggier2011secrecy}
F.~Oggier and B.~Hassibi, ``The secrecy capacity of the mimo wiretap channel,'' \emph{IEEE Transactions on Information Theory}, vol.~57, no.~8, pp. 4961--4972, Aug. 2011.

\bibitem{10556582}
B.~Li, X.~Wang, and F.~Fang, ``Maximizing the value of service provisioning in multi-user {ISAC} systems through fairness guaranteed collaborative resource allocation,'' \emph{IEEE JSAC}, vol.~42, no.~9, pp. 2243--2258, Jun. 2024.

\bibitem{5559377}
J.~D. Griffin and G.~D. Durgin, ``Multipath fading measurements at 5.8 {GHz} for backscatter tags with multiple antennas,'' \emph{IEEE Transactions on Antennas and Propagation}, vol.~58, no.~11, pp. 3693--3700, Aug, 2010.

\bibitem{tse2005fundamentals}
D.~Tse and P.~Viswanath, \emph{Fundamentals of Wireless Communication}.\hskip 1em plus 0.5em minus 0.4em\relax Cambridge, U.K.: Cambridge University Press, 2005.

\bibitem{10599122}
H.~Huang, L.~Dai, H.~Zhang, C.~Zhang, Z.~Tian, Y.~Cai, A.~L. Swindlehurst, and Z.~Han, ``Disco might not be funky: Random intelligent reflective surface configurations that attack,'' \emph{IEEE Wireless Communications}, vol.~31, no.~5, pp. 76--82, Jul. 2024.

\bibitem{watson1922treatise}
G.~N. Watson, \emph{A treatise on the theory of Bessel functions}.\hskip 1em plus 0.5em minus 0.4em\relax The University Press, 1922, vol.~3.

\bibitem{gradshteyn2014table}
I.~S. Gradshteyn and I.~M. Ryzhik, \emph{Table of integrals, series, and products}.\hskip 1em plus 0.5em minus 0.4em\relax Academic press, 2014.

\bibitem{10271138}
X.~Li, Y.~Pei, X.~Yue, Y.~Liu, and Z.~Ding, ``Secure communication of active {RIS} assisted {NOMA} networks,'' \emph{IEEE Transactions on Wireless Communications}, vol.~23, no.~5, pp. 4489--4503, Oct. 2024.

\bibitem{dickey1983multiple}
J.~M. Dickey, ``Multiple hypergeometric functions: Probabilistic interpretations and statistical uses,'' \emph{Journal of the American Statistical Association}, vol.~78, no. 383, pp. 628--637, Oct. 1983.

\end{thebibliography}

\end{document}